%% file: main.tex
\documentclass[sigplan,10pt,nonacm=true]{acmart}
\renewcommand\footnotetextcopyrightpermission[1]{}
\AtBeginDocument{%
 }

\input{macros}
\setcopyright{none}

\begin{document}


\title{Scalable datacenter replication with mostly-synchronous consensus on hardware}

\author{Davide Rovelli}
\email{roveld@usi.ch}
\orcid{0000-0002-6881-3126}
\affiliation{%
  \institution{Universit\`a della Svizzera Italiana (USI)}
  \city{Lugano}
  \country{Switzerland}
}

\author{Philipp Berdesinski}
\email{philipp@berdesinski.dev}
\orcid{0009-0009-2155-2633}
\affiliation{%
  \institution{turbalance}
  \city{Heidelberg}
  \country{Germany}
}

\author{Rodrigo Otoni}
\email{r.b.otoni@rug.nl}
\orcid{0000-0003-1097-2367}
\affiliation{%
  \institution{University of Groningen}
  \city{Groningen}
  \country{Netherlands}
}

\author{Patrick Eugster}
\email{eugstp@usi.ch}
\orcid{0000-0003-3864-9078}
\affiliation{%
  \institution{Universit\`a della Svizzera Italiana (USI)}
  \city{Lugano}
  \country{Switzerland}
}
\renewcommand{\shortauthors}{D. Rovelli et al.}

\input{sections/abstract}


\maketitle

\input{sections/introduction}
\input{sections/model}
\input{sections/design}
\input{sections/algorithm}

\input{sections/implementation}
\input{sections/evaluation}
\input{sections/relatedWork}
\input{sections/conclusion}


\bibliographystyle{ACM-Reference-Format}
\bibliography{references_v2}

\iftoggle{withAppendix}{%
    \appendix
    \pagebreak
    \input{sections/appendix}

    }{%
}

\end{document}

%% file: macros.tex
\usepackage{acronym}
\usepackage{enumitem}
\usepackage{comment}
\usepackage{hyperref}

\usepackage[ruled,vlined]{algorithm2e}
\SetKwBlock{To}{to}{}
\SetKwBlock{Upon}{upon}{}
\SetKwBlock{Uses}{Uses}{}
\SetKw{Kwand}{and}
\SetKw{Kwor}{or}
\DontPrintSemicolon

\newcommand{\algobox}[4]{\tikz[baseline=(box.base)]{\node[draw=#1, fill=#2, #3, inner sep=1pt] (box) {#4};}}
\newcommand{\algoboxAE}[1]{\algobox{purple}{pink!80}{solid}{#1}}
\newcommand{\algoboxLK}[1]{\algobox{blue}{blue!20}{dotted}{#1}}
\newcommand{\algoboxDQ}[1]{\algobox{red}{yellow!40}{dashed}{#1}}

\usepackage{amsthm}

\theoremstyle{definition}

\newtheorem{observation}{Observation}[section]

\newtheorem{assumption}{Assumption}[section]

\newtheorem{lemma}{Lemma}[section]

\newtheorem{theorem}{Theorem}[section]

\usepackage{caption}
\usepackage{subcaption}

\usepackage{tabularx}
\usepackage{booktabs}
\usepackage{makecell}
\usepackage[table]{xcolor}
\usepackage{colortbl}
\usepackage{array}
\definecolor{lighttgray}{gray}{0.9}
\newcolumntype{L}[1]{>{\raggedright\arraybackslash}p{#1}}
\newcolumntype{g}[1]{>{\columncolor{lightgray}\raggedright\arraybackslash}p{#1}}

\newtoggle{withAppendix}
\toggletrue{withAppendix}

\newtoggle{ArxivVersion}
\toggletrue{ArxivVersion}

\newacro{name}[scarHW]{scalable replication in-hardware}
\newacro{sync}[sync]{synchronous}
\newacro{fast}[timely]{timely}
\newacro{sdcl}[SDCL]{\acl{sync} datacenter layer}
\newacro{async}[async]{asynchronous}
\newacro{nic}[NIC]{network interface controller}
\newacro{cpu}[CPU]{central processing unit}
\newacro{bft}[BFT]{Byzantine fault tolerance}
\newacro{sdn}[SDN]{software-defined network}
\newacro{te}[TE]{traffic engineering}
\newacro{utb}[UTB]{upper time bound}
\newacro{mtu}[MTU]{maximum transmission unit}
\newacro{ipi}[IPI]{inter-processor interrupt}
\newacro{fd}[FD]{failure detector}
\newacro{rdma}[RDMA]{remote direct memory access}
\newacro{tor}[TOR]{top-of-rack}
\newacro{roce}[RoCE]{RDMA over converged Ethernet}
\newacro{dpdk}[DPDK]{data plane development kit}
\newacro{rpc}[RPC]{remote procedure call}
\newacro{vm}[VM]{virtual machine}
\newacro{os}[OS]{operating system}
\newacro{lkm}[LKM]{Linux kernel module}
\newacro{xdp}[XDP]{express data path}
\newacro{irq}[IRQ]{interrupt request}
\newacro{bpf}[BPF]{Berkeley packet filter}
\newacro{ebpf}[eBPF]{extended Berkeley packet filter}
\newacro{rcu}[RCU]{read-copy-update}
\newacro{skb}[SKB]{socket kernel buffer}
\newacro{zab}[ZAB]{Zookeeper Atomic Broadcast}
\newacro{smr}[SMR]{state machine replication}
\newacro{mpi}[MPI]{message passing interface}
\newacro{snic}[smartNIC]{smart network interface controller}
\newacro{ipu}[IPU]{infrastructure processing unit}
\newacro{soc}[SoC]{system on a chip}
\newacro{gst}[GST]{global stabilization time}
\newacro{pcie}[PCIe]{peripheral component interconnect express}
\newacro{fnic}[FPGA smartNIC]{FPGA smartNIC}
\newacro{mpps}[Mpps]{million packets per second}
\newacro{redisnc}[Redis-NC]{Redis-NC}
\newacro{zookeepernc}[Zookeeper-NC]{Zookeeper-NC}
\newacro{tob}[TOB]{total order broadcast}
\newacro{ptp}[PTP]{precision time protocol}
\newacro{rtt}[RTT]{round-trip time}
\newacro{conic-l}[scarHW-L]{scarHW leader}
\newacro{conic-lb}[scarHW-B]{scarHW load-balanced}
\newacro{conic-us}[scarHW-u]{scarHW microservice}
\newacro{fpga}[FPGA]{field programmable gate array}
\newacro{xrt}[XRT]{Xilinx runtime}
\newacro{axi}[AXI]{advanced extensible interface}
\newacro{axilite}[AXI4-lite]{AXI4-lite}
\newacro{axifull}[AXI4]{AXI4 full}
\newacro{axistream}[AXI-stream]{AXI-stream}
\newacro{mm2s}[MM2S]{memory-to-stream}
\newacro{s2mm}[S2MM]{stream-to-memory}
\newacro{ckc}[POPUC]{popular uniform consensus}
\newacro{ckcuc}[POPUC
]{popular uniform consensus}
\newacro{ckcucmi}[POPUC\textsubscript{mi}]{popular uniform consensus for multiple instances}
\newacro{ckcfv}[POPUC-fv]{popular uniform consensus with full value set relay}
\newacro{ckctob}[POPTOB]{popular total order broadcast}
\newacro{cc}[co-consensus]{collaborative consensus}
\newacro{vr}[VR]{Viewstamped Replication}

\newcommand{\tool}[1]{\textsc{#1}}

\newcommand{\apalache}{\tool{Apalache}}

\newcommand{\prop}[1]{\textnormal{\textit{#1}}}
\newcommand{\prim}[1]{\textsc{#1}}
\newcommand{\algprim}[3][]{\prim{#2}\ensuremath{_{{\textup{#1}}}#3}}
\newcommand{\propprim}[1]{\prim{#1}}
\newcommand{\algval}[1]{\prim{#1}}

\usepackage[disable]{swnotes}
\newcommand{\dvrnote}[1]{\annote{#1}{dvr}{yellow}}
\newcommand{\pnote}[1]{
\annote{#1}{p@}{orange}
}

\newcommand{\rnote}[1]{\annote{#1}{ro}{green}}

\acused{os}

\newcommand{\mymin}[2]{\ensuremath{\min^{#1}_{#2}}}

%% file: sections/abstract.tex
\begin{abstract}

Consistent replication of data among distributed processes -- a task involving the well-known consensus problem -- is notoriously expensive and hard to scale, affecting especially datacenter services with stringent performance requirements.
To mitigate this problem, we introduce \emph{\ac{name}}: a network card design that improves throughput and latency of consistent replication even when increasing the number of replicas, whereas current systems operate at a small scale or with relaxed consistency guarantees. 
At the heart of \ac{name} is our novel \acs{ckcuc} consensus algorithm, implemented in an \acs{fpga} \acs{snic} to take full advantage of the ``mostly synchronous'' behavior of programmable network devices in the datacenter. 
Unlike widely-adopted ``mostly asynchronous'' coordination protocols such as Paxos or leaderless alternatives, \acs{ckcuc} implements a generalized variant of consensus dubbed \emph{\acl{cc}} which allows for several simultaneous decisions, achieving great scalability without compromising availability. 
\acs{ckcuc} preserves safety guarantees in the presence of process crash-stop and message send/receive omission failures (capturing incidental asynchrony) and has been formally specified and verified in TLA$^+$.
Our FPGA prototype improves throughput and latency of widely-used services Redis and Zookeeper by up to two orders of magnitude compared to the state of the art.
\ac{name}-based services also achieve zero downtime upon failure of a minority of replicas, offering a highly-robust, wire-speed, scalable replication system.

\end{abstract}

%% file: sections/introduction.tex
\section{Introduction} \label{sec:introduction}


\Ac{smr} is a mechanism at the heart of countless highly-available services which need to ensure a consistent state across 
nodes. 
These services include heavily-demanded
online applications running in datacenters, such as for coordination
~\cite{zookeeper}, 
event streaming
~\cite{kreps2011kafka}, 
data 
storage~\cite{taft2020cockroachdb, zhou2021pull}, key-value storage~\cite{etcd, redisraft}, memory disaggregation
~\cite{aceso_disaggmem_Hu_SOSP24, hydra_disaggmem_Lee_FAST22}, and resilient large language model training~\cite{llm_async_checkpoint_Mauriya_HPDC24}.
Because they form the backbone of critical infrastructure, \ac{smr} implementations need to be robust to failures. This leads to the long-standing challenge of 
achieving fault-tolerant coordination among distributed processes, i.e., 
achieving consensus, \emph{efficiently}.

\paragraph{Constraints of asynchronous protocols.}
Current coordination services still predominantly assume the ``mostly asynchronous'' partial synchrony model~\cite{DWLYST88} in which, prior to a (unknown) \ac{gst}, 
computing and communication times are unbounded due to arbitrary resource contention on end nodes and in the network.
This assumption introduces several costs in protocol design
, which the community has embraced as the inevitable price to pay for fault tolerance. 
The vast majority of state-of-the-art \ac{smr} services thus rely on variants of popular (mostly) asynchronous protocols Paxos~\cite{paxos} and Raft~\cite{raft}. Recent works propose 
optimizations 
using modern datacenter technologies such as programmable switches~\cite{Dang2015NetPaxos,Li2016NOPaxos}, \acp{snic} with \acp{fpga}~\cite{waverunner, consensus_in_a_box, nanopu, paxos_in_the_nic},\acused{nic} \ac{rdma}~\cite{mu, p4ce_dulong_icdcs24, ukharon}, and \ac{os}-bypass frameworks~\cite{electrode, xlane}. 
While 
boosting performance,  
most such solutions inherit 
the fundamental limitations of the underlying system model and algorithms. 
Crucially, 
by funneling replication requests through a \emph{leader process}, they induce three major constraints:

\begin{enumerate}[label=\textbf{C\arabic*}.,ref=\textbf{C\arabic*}]
    \item\label{itm:bottleneck} \textbf{Single-process bottleneck:} 
    the leader processes client requests, forwards them to all replicas and receives acknowledgments from a majority of them. 
    Resources available at the leader significantly constrain the performance of the whole cluster~\cite{p4ce_dulong_icdcs24, paxos_variants_perf_murat_SIGMOD19, wpaxos_ailijiang_murat_2020}.
    \item\label{itm:ft} \textbf{Low scalability and fault tolerance:} The leader bottleneck incurs particularly rapid performance degradation as the number of replicas increases. System architects mitigate this issue by restricting replication to few  nodes~\cite{zookeeper_admin, kubernetes_etcd_operations}. This, however, strongly limits the number of tolerated failures 
    given that a correct majority quorum is needed for liveness.
    %
    
    \item\label{itm:downtime} \textbf{Service downtime:} During leader election or cluster reconfiguration the system cannot process requests. 
\end{enumerate}

\paragraph{Multi-leader and leaderless SMR}
Other classical consensus protocols on top of mostly asynchronous models, e.g., rotating coordinator~\cite{chandra1996unreliable}, may rely on leaders less (explicitly),  
but exhibit other constraints~\cite{paxos_vs_rotatingleader_urban_schiper_2004,mostefaouiraynal2001leader}.
More popular solutions to circumvent \ref{itm:bottleneck} and \ref{itm:ft} include multi-leader versions of Paxos which use small replication groups/shards with consensus mechanisms on top to coordinate group management~\cite{wpaxos_ailijiang_murat_2020,cockroach_db,VPaxos_LamportMZ09}, and others attempts at decentralization~\cite{Marandi2010RingPaxos,EgalitarianPaxos_Moraru2013,Derecho_Jha-Birman2019}. Alas, such approaches pay the price of loosening consistency guarantees (e.g., sharding implies partial instead of total order), or degrading \ref{itm:downtime} and 
response latency in the face of contention~\cite{EgalitarianPaxos_Moraru2013} or stragglers~\cite{Derecho_Jha-Birman2019, ChoraSync_LiArxiv25}. 

\paragraph{Mostly synchronous datacenters.}
We instead propose to tackle \ref{itm:bottleneck}-\ref{itm:downtime} more fundamentally by building a solution ground-up on top of a different system model. We start from the observation that programmable datacenter networks can provide \emph{stable, low time bounds}
~\cite{SpecPaxos_PortsNSDI2015, The_Synchronous_Datacenter_Yang_2019
}, further supported by the rise of ultra-low latency network protocols claiming bounded datacenter communication latency~\cite{fastpass_Perry_sigcomm14,xlane}. Moreover, now widely-available \acp{snic} can significantly strengthen synchrony in practice through fast, uninterrupted packet processing at the endhost~\cite{Nanoconsensus_RovelliSoCC25}.
However, though rare, synchrony violations have to be taken into account, as a single violation can otherwise lead to inconsistencies. 


\paragraph{ScarHW: \acl{cc} on \ac{fpga}.}
In light of these insights, we propose
the \emph{\acf{name}}: a network card design providing a robust and efficient \ac{smr} service offloaded to an \ac{fpga} 
in order to leverage the predictability and speed of datacenter networks.
At the heart of \ac{name} is the  \ac{ckcuc}: a \emph{novel consensus algorithm} which 
exploits ``mostly \emph{synchronous}'' systems for performance while ensuring fault tolerance in the presence of process crash-stop as well as \emph{message send/receive omission failures} to capture any 
incidental late or lost messages. 
By relying on a correct majority quorum, \ac{ckcuc} has the same message complexity, safety guarantees, and actual fault tolerance as Paxos (including when a majority fails, cf. \autoref{tab:comp}), but has 
crucial advantages. %
Most importantly, \ac{ckcuc} is  
leaderless, solving a variant of consensus which we call \emph{\ac{cc}}, that allows several  \emph{simultaneous} decisions from multiple propositions.  
This approach allows client requests to be 
issued 
across replicas 
without impacting latency, unlike classic batching strategies (which \ac{name} can also use). By harnessing  
\emph{load distribution} and hardware acceleration, \ac{name}-based services thus allow the increase of fault tolerance by adding more replicas while also \emph{maintaining or even improving throughput and latency of consistent replication}, overcoming \ref{itm:bottleneck} and \ref{itm:ft}. However, unlike asynchronous leaderless protocols, \ac{name} does not compromise latency, does not require partitioning/sharding (which would make it even more scalable) and achieves \emph{negligible service downtime} up to $t = \lfloor\frac{N-1}{2}\rfloor$ failures, eliminating \ref{itm:downtime}. 
In short, by leveraging hardware support with our system/failure model, \ac{name} reaps all the 1. pros of synchrony without suffering any 2. cons often attributed with it: 1. inherently leaderless coordination without 2. rigid lock-step execution (early termination in absence of failures) or conservatively high latency bounds (stable maximum interaction latencies below average of asynchronous approaches).

\begin{table*}[t!]
{ 
\small
\caption{
\ac{name} vs state-of-the art \ac{smr} approaches 
w.r.t.: 
consensus variant, system and 
failure models,  model (and mechanisms in algorithms) handling  
disruptions, 
failure bound and guarantees.   
$|Q|\geq t+1 =\lceil\frac{N+1}{2}\rceil$, \mymin{f}{t}=$\min(f+2,t+1)$.
}\label{tab:comp}

\setlength{\tabcolsep}{
.85mm}
\begin{tabular}{c|c|cc|ccc|c|cc}
\toprule 
{\bf Approach} & {\bf Consensus} & \multicolumn{2}{c|}{\bf Model} &\multicolumn{3}{c|}{\bf Model handling disruption} & {\bf Quorum} & \multicolumn{2}{c}{\bf Guarantees} \\
& {\bf variant} & \makecell{\bf System} & \makecell{\bf Failure} & \makecell{\bf Process fault} & \bf Delay &\makecell{\bf Message loss} & \bf size & \bf Safety & \bf Liveness  \\
\midrule

\makecell{\bf Standard\\\bf\ac{smr}} & \makecell{Selective\\(1 decision)} & 
\makecell{Partially\\synchronous} & \makecell{Process\\crash-stop} & \makecell{Failure\\(quorums)} & 
\makecell{System\\(quorums)} & 
\makecell{System\\(quorums/\\resend)} & 
\makecell{$t+1$ (timely $>$\\ GST)  commun.\\ processes}  & 
\makecell{Always\\ (blocking\\ with $f>t$)}
&  \makecell{$>$GST}
\vspace{1mm}\\

\textbf{\ac{name}} & \makecell{Collaborative\\($|Q|$ decisions)} &
\makecell{Synchronous}&
\makecell{Process\\crash-stop,\\ omission}&
\makecell{Failure\\(quorums)}&
\makecell{Failure\\(drop\\cf. loss)}&
\makecell{Failure\\(quorums)}&
\makecell{$t+1$ timely\\communicating \\processes}
& 
\makecell{Always\\ (blocking\\ with $f>t$)} &
\makecell{\mymin{f}{t}\\rounds} \\
\bottomrule
\end{tabular}
}
\end{table*}

Our prototype -- implemented on an AMD/Xilinx Alveo U50~\cite{AMD_Alveo_U50} \ac{snic} -- can saturate a 100Gbps network with worst-case 5.4$\mu$s response latency. With more than 3 replicas, \ac{name} delivers up to 2.5$\times$ higher replication 
throughput than Raft-based Waverunner~\cite{waverunner}, the fastest (to 
our knowledge) hardware \ac{smr} service. 
The full potential of \ac{name} emerges when integrated into microservices: when scaling up to 59
nodes, \ac{name} improves Redis by up to two orders of magnitude in both throughput and latency compared to NOPaxos~\cite{Li2016NOPaxos}.

\paragraph{Contributions.}
In summary, this paper

\begin{itemize}

    \item motivates the mostly synchronous system model with process crash-stop \emph{and} send/receive omission failures as a better match for high performance datacenter services than traditional mostly asynchronous models (or synchronous models) with only crash failures; (\autoref{sec:model});  
    
    \item proposes the design of \ac{name}: a novel, flexible network card design which harnesses the power of distribution and hardware acceleration (\autoref{sec:design});
    
    \item specifies \ac{ckcuc}: \ac{name}'s novel leaderless protocol solving \ac{cc}, whose correctness has been formally verified via TLA$^+$ model checking (\autoref{sec:algorithm});

    \item presents the implementation of our prototype on the  AMD/Xilinx Alveo U50~\cite{AMD_Alveo_U50} \ac{fpga} \ac{snic} (\autoref{sec:implementation});
    
    \item empirically evaluates \ac{name} in terms of raw performance, impact for real-world applications, and fault tolerance (\autoref{sec:eval}).  
    We show that \ac{name}-based services are not affected by traditional bottlenecks (\ref{itm:bottleneck}, \ref{itm:ft}, \ref{itm:downtime}), achieving unprecedented performance.

\end{itemize}

\autoref{sec:relatedWork} presents related work. \autoref{sec:conclusion} draws conclusions. 
\iftoggle{withAppendix}{
\autoref{apx:algorithms} details 
\ac{ckcuc} and algorithms for \ac{smr} built on top of it, and argues for correctness.
}
{
Supplementary material details \ac{ckcuc} and algorithms for \ac{smr} built on top of it, and argues for correctness.
}

%% file: sections/model.tex
\section{System and Failure Model}\label{sec:model}


\subsection{Specification}
We consider a message-passing distributed system with a set $\Pi$ of $N = |\Pi|$ processes, also called replicas. 
We consider mostly synchronous systems, assuming two bounds   $\Delta_I$ and $\Delta_D$  \emph{for performance} only, not 
correctness.
$\Delta_I$ bounds most interaction latency between any two processes $p_i, p_j \in \Pi$ , i.e., communication between processes together with end-to-end 
processing happens most of the time within $\Delta_I$. 
$\Delta_D$\dvrnote{we could remove these Deltas as they are not used later on to define round time anymore} bounds clock drifts between two processes, which is easily achieved in practice and assumed by state-of-the-art datacenter coordination works, including for correctness, e.g.~\cite{ukharon}.

Processes can suffer from crash-stop failures, and omission failures 
defined below.

\begin{description}[font=\textit]
    \item[Omission failure:] 
    Processes fail to send or receive messages. Omission failures capture messages dropped (e.g., due to network faults) or  with  latency 
    $>\Delta_I$ (which are discarded). 
\end{description}


A process is \textit{faulty} if it commits a crash-stop failure or at least one omission failure, otherwise it is \textit{correct}. 
We assume a majority quorum of \emph{correct} processes exists such that $|Q| \ge t + 1$ and $N = 2t + 1$, where $t$ is the maximum numb/er of faulty processes. We use $f$ to denote the number of \emph{actually} faulty processes.\dvrnote{mention that other version loosen the assumption to a sync quorum?}
The majority quorum assumption is proven to be necessary to solve uniform consensus in synchronous systems with omission failures~\cite{sync_consensus_summary_raynal, omission_sync_uniform_consensus:SPAA04}. 
When $f > t$, our consensus protocol remains safe but not live (i.e., it may block), as detailed \mbox{in~\autoref{sec:algorithm}}. 

\subsection{Paradigm shift to better model datacenters}\label{sec:shift}

Apart from a small number of algorithms explored more in theory~\cite{CF99,agreement_comm_faults_toueg, sync_agreement_ubiquitous_faults_santoro, sync_consensus_summary_raynal, sync_consensus_hybrid_faults_biely, kset_agreement_sync_omission_parvedy_raynal}, distributed coordination algorithms developed over the past 20+ years largely focused on crash-stop failures \emph{without omissions} (unless as part of Byzantine failures) in the partial synchrony model based on a conservative view: interactions between distributed nodes can be highly unreliable because endhost processing times can be unbounded due to resource contention, and networks can arbitrarily delay or drop packets. 
This leads to the common \emph{mis}belief that assuming synchrony in the presence of %
crash-stop  
failures \emph{must} be unsafe, since a single delayed packet can be misinterpreted as a failure and hamper consistency. 
Yet, solutions including \emph{omission failures} in the  \emph{failure model}, like \ac{name}, 
can handle such synchrony violations (see \autoref{tab:comp} for comparison with the standard \ac{smr} model).

Such  solutions could be deemed impractical as \emph{frequent} packet delays or drops captured as failures could quickly degrade the availability of a cluster~\cite{sync_consensus_hybrid_faults_biely}. While this may be the case in 
some systems, recent advances seriously challenge this common assumption \emph{for datacenters}~\cite{SpecPaxos_PortsNSDI2015,The_Synchronous_Datacenter_Yang_2019, DCCs_RovelliEugster_DSN25}. 
Modern programmable datacenter hardware, e.g., \ac{fpga} \acp{snic}, namely not only allows for very fast access to the network (as exploited to achieve ultra-low tail latency in coordination~\cite{waverunner, consensus_in_a_box, nanopu, ukharon, mu}), but also make communication and computation times very deterministic.   
Several state-of-the-art datacenter coordination services recently \emph{enforce} synchrony for a subset of a datacenter, claiming reliable upper bounds on latency achieved through isolated  \ac{os} packet processing pipelines on commodity \acp{nic} and network redundancy to overcome packet losses~\cite{fide,Nanoconsensus_RovelliSoCC25}. However, even if strongly reducing the probability of network failures being witnessed, these can not be fully ruled out.

We propose a fundamental model shift for datacenter applications, based on the intuition that failure models (rather than system models) should focus on rare cases like message losses and latency outliers.   
We use the synchronous system model to better represent interaction reliability among \acp{name} with (in practice few) glitches captured as omission failures, rather than assuming asynchrony, pessimistically, as the norm, or then perfect synchrony.
This allows \ac{name} to reap all the 1. pros of synchrony without suffering any 2. cons often attributed with it: 1. inherently leaderless coordination without 2. rigid lock-step execution or conservatively high latency bounds as substantiated in the following.

\subsection{How practical 
is ``mostly synchronous''?
}
\label{sec:practical_sync}

We offload \ac{name} to \ac{fpga} \acp{snic} to fully leverage the predictability and speed of datacenter networks. 
To assert the synchrony assumptions of our model, we ran initial evaluations to quantify the reliability of modern network hardware.
The microbenchmark consists in  
a simple ping-pong protocol between two nodes connected via a \acl{tor} switch (see \autoref{sec:eval_setup} for hardware specification). 
We compute the one-way latency for every packet as half of the \ac{rtt}. We ran our benchmark in a production datacenter of a major cloud provider\footnote{Anomymized for double-blind submission.} for 40 consecutive days collecting measurements for a total of 115 billion packets. 
For \ac{name}, we sent packets at a constant throughput of 10KHz and measured the latency including the hardware packet processing pipeline. We also evaluate latency of two modern high-performance 
networking 
solutions: \ac{rdma} unreliable connection (UC) and \ac{xdp} AF\_XDP sockets through the same ping-pong test, but using a much smaller sample base of 10 million packets (advantaging them over \ac{name}). We add UDP as a baseline. Throughout 
the experiment, we periodically applied up to 100\% network traffic and CPU load using \texttt{iperf3} and \texttt{stress-ng}. We manually reserved high-priority queues on 
switches for the benchmark traffic and left best-effort queues for the other traffic to replicate the effect of low-latency communication protocols~\cite{fastpass_Perry_sigcomm14,qjump,xlane
}.

\begin{figure}[t]
    \centering
    \includegraphics[width=\linewidth]{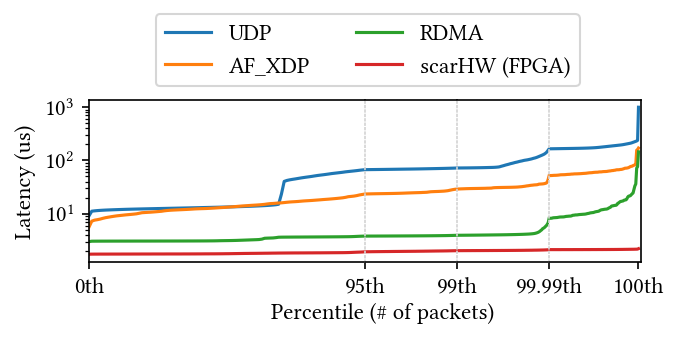}

    \caption{Latency of different network stacks at the endhost.}
    \label{fig:stable_latency}
\end{figure}

\autoref{fig:stable_latency} shows that pushing logic to dedicated processing resources close to the wire leads to extremely stable latency of a few $\mu$s while approaches involving software show spikes towards the tail. 
\ac{name} interactions register no outlier over 4.41$\mu$s and no packet drops, suggesting that upper time bounds can be confidently established with $\mu$s margins. \ac{name}'s maximum interaction latency measured is even below the minimum of the other approaches, including \ac{rdma}.

The use of synchrony as default in practice is further motivated by a wide range of available options for ultra-reliable network communication which can 
enable bounded-latency communication, allowing datacenter operators to navigate the tradeoff between network utilization and reliable traffic~\cite{
The_Synchronous_Datacenter_Yang_2019,Nanoconsensus_RovelliSoCC25}. 
Moreover, while reliable communication requires resource reservation in \ac{name}, deployment of targeted critical services in practice with \emph{current} solutions includes comparable yet error-prone \emph{manual} reservation of resources~\cite{zookeeper_admin, etcd}, 
e.g., bandwidth or priority queues, to decrease the possibility of congestion for prime performance.  

In the following we use high-priority queues reserved on switches (like for \autoref{fig:stable_latency}), 
to showcase how this simple mechanism can effectively support robust synchrony  despite concurrent regular traffic. However, it is important to note that \emph{traffic prioritization is not required for safety}. \ac{name} guarantees consistent replication despite any number of omission failures, including delays, packet drops, and network faults, allowing fine-tuning of upper time bounds to optimize performance without compromising safety.
Although asynchronous models are more tolerant of short transient omission failures (which are rare in well-provisioned environments), \ac{name} offers the distinct advantage of providing hardware-precise notifications of unexpected delays. This signaling can be leveraged for congestion control, rapid response to network partitions, or detecting a cluster stall when more than $t$ failures occur (cf. \autoref{sec:algo_comparison}). Combining \ac{name} with further techniques like traffic engineering as in prior works~\cite{xlane,fide} 
could further improve performance. 

%% file: sections/design.tex
\section{Design} \label{sec:design}



\subsection{Architecture}\label{sec:architecture}

\begin{figure}
    \centering
    \begin{subfigure}{\linewidth}
        \includegraphics[width=\linewidth]{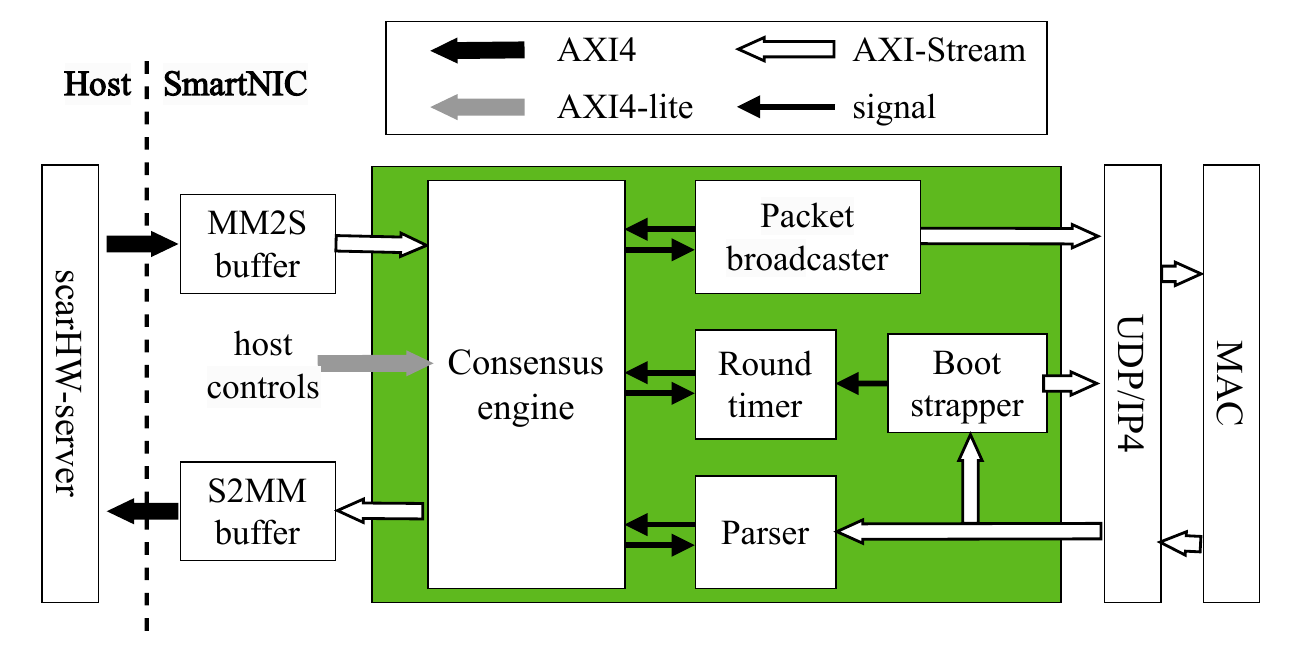}
        \caption{\ac{name} vanilla}
        \label{fig:architecture_diagram_vanilla}
    \end{subfigure}
    \begin{subfigure}{\linewidth}
        \includegraphics[width=\linewidth]{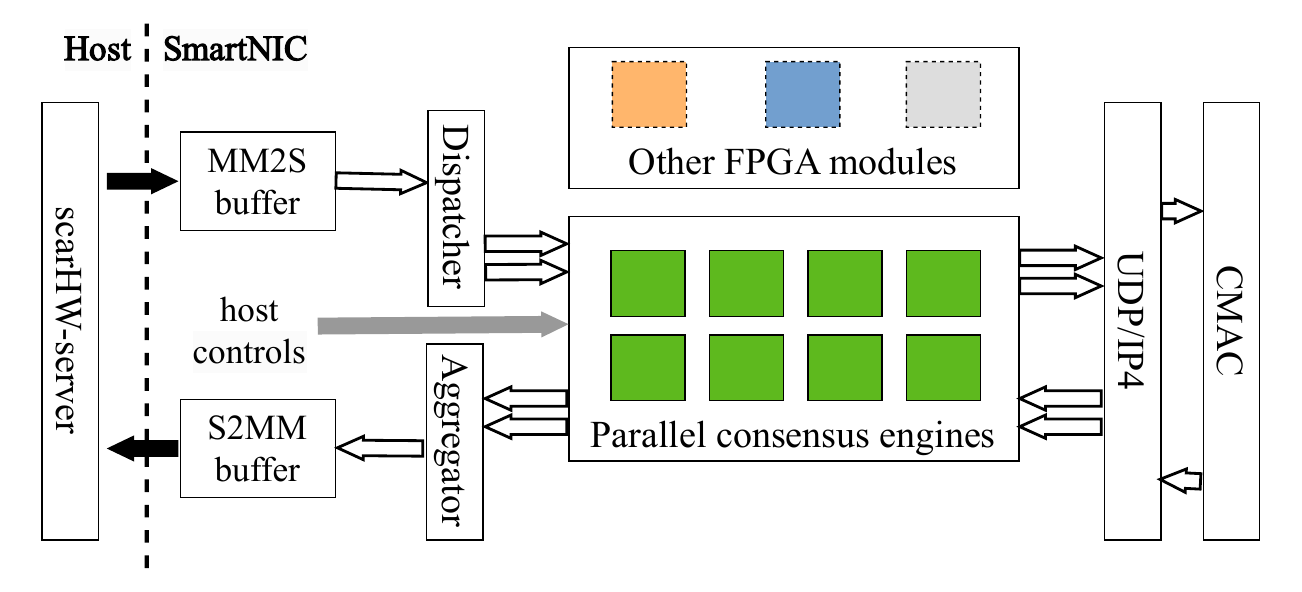}
        \caption{\ac{name} pipelined}
        \label{fig:architecture_diagram_pipelined}
    \end{subfigure}

    \caption{Architecture of \ac{name} variants. Hardware modules use \ac{axi} interconnects.}
    \label{fig:architecture_diagram}
\end{figure}

\autoref{fig:architecture_diagram} outlines the 
architecture of \ac{name} in two variants: \emph{vanilla} (\ref{fig:architecture_diagram_vanilla}), containing 
one replication unit (green box), and \emph{pipelined} (\ref{fig:architecture_diagram_pipelined}), aggregating multiple 
units.
Every \ac{snic} 
running \ac{name} is a replica executing the same logic. 
The brain of the system is the consensus engine module which implements the \ac{ckcuc} algorithm (cf. \autoref{alg:ckcuc}).  In short, the module takes arbitrary values from one or more applications as inputs, proposes them with a unique \textsc{id} 
for the current consensus instance, and outputs consensus decisions. 
The consensus engine operates in synchronous rounds in which all replicas disseminate messages to each other. In a failure-free execution, all replicas \emph{propose values simultaneously} and, after two rounds, output an ordered array of all proposed values as the consensus decision. This key feature of \ac{name} enables load distribution, which we will discuss in later sections. The engine keeps track of packets it receives in every round and flags a remote replica as faulty if its message is not received within the expected timeframe, i.e., it observes an omission failure. To avoid false positives, replicas can send empty \emph{heartbeat} values in absence of pending proposals. Otherwise, replicas can choose to not send a heartbeat on purpose to scale up/down, or allowing re-configuration during ``quiet'' periods with few requests.
In the presence of up to $t$ failures, a decision can take between 2 and $t$ rounds, and replicas can decide in different rounds. When a replica observes more than $t$ failures, e.g., due to 
a network outage, the consensus engine self-declares itself faulty and quits execution.

To avoid costly retransmissions, which would compromise our algorithm's performance, \ac{name} replicas communicate via UDP over Ethernet through the UDP/IP4 and MAC modules.  
The communication between the host and the consensus engine involves two memory-mapped buffers separated for input proposals and consensus decisions using \ac{mm2s} interfaces which bridge the stream nature of the UDP/IPv4 layer and the memory-oriented structure of host-\ac{fpga} transfers. 
Host processes interact with \ac{name} via \texttt{scarHW-server}: a thin software layer that can handle client requests, seamlessly integrate with applications, and configure hardware settings.
The packet broadcaster turns the payload supplied by the consensus engine module into a suitable stream for the UDP/IP4 layer, while the parser performs the opposite function. The bootstrapper module is of particular relevance since it allows the synchronization of the replicas before the start of the first consensus instance by adequately triggering the round timer after an initial message exchange.
After this bootstrap sequence, the consensus engine is normally set to run \emph{back-to-back consensus instances in sequence} in order to achieve \ac{smr} (cf. \autoref{sec:algo_comparison}) while maintaining synchronization. The consensus engine can adjust the round timer dynamically to correct timer drift, in a similar fashion as clock synchronization methods~\cite{ptp_ieee1588}.

The replication throughput of \ac{name} vanilla is limited by the round time, i.e., worst-case latency of the setup plus a safety margin, since replicas have to wait for the decision from a consensus instance before starting the next one. \ac{name} pipelined uses multiple consensus engines, timers and packet handling modules in parallel (grouped into a green module in \autoref{fig:architecture_diagram_vanilla}) to maximize decision rate. Individual green modules in different \acp{name} have matching \textsc{id}s and form a replication group to avoid overlapping protocol states. Dispatcher and aggregator modules provide basic functionalities to respectively parallelize and serialize consensus decisions. Serialization sorts parallel decisions by green module \textsc{id}s to ensure deterministic order across different replicas.
Dispatcher and aggregator settings allow a single application to use different replication groups to increase throughput and also multiple applications to use different replication groups as unrelated \ac{smr} instances.

\subsection{Deployment modes}

\begin{figure}
    \centering
    \includegraphics[width=\linewidth]{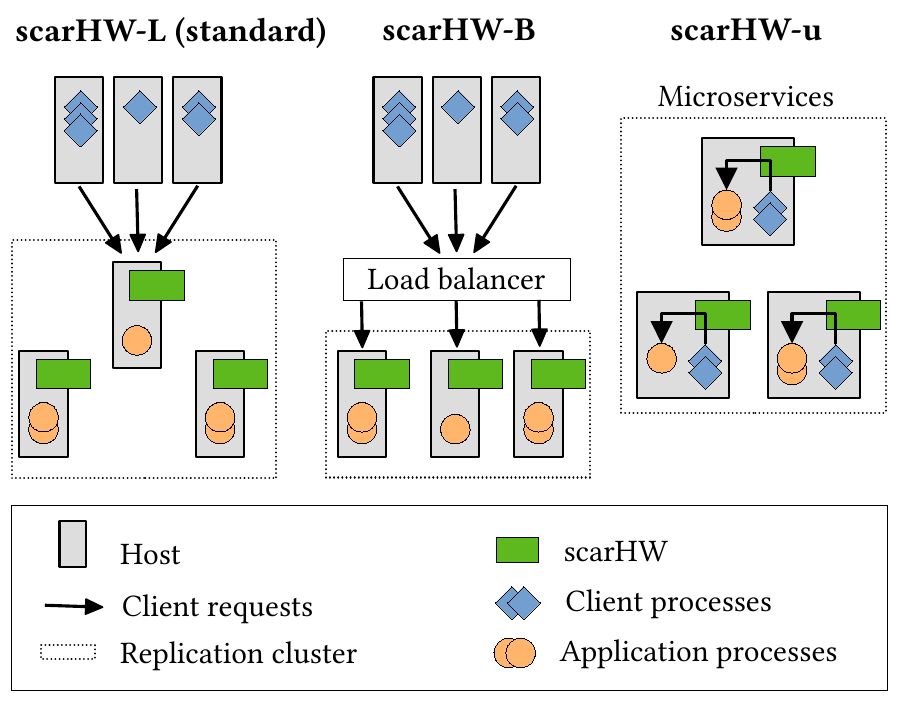}
    \Description{Diagram showing possible deployments for replication clusters using the system.}
    \caption{Possible deployments for replication clusters using \ac{name}. Classical leader-based consensus solutions allow only one clients-service interaction (left) while \ac{name} enables two additional configurations (center and right).}
    \label{fig:deployment_modes}
\end{figure}

Consistently-replicated services running on top of \ac{name} can leverage the power of \emph{load distribution} since all replicas can propose values \emph{simultaneously}. This allows programmers to choose among the different deployment options outlined in \autoref{fig:deployment_modes}. 
In \ac{conic-l} mode, all clients forward requests to one node which in turn forwards requests to the replication cluster. This option is the least beneficial and is included only to compare \ac{name} against other leader-based approaches. 
\ac{conic-lb} mode mitigates the bottleneck by using a load balancer to dispatch the requests to all remote replicas, allowing for custom load shaping strategies. Lastly, \ac{conic-us} 
is tailored to microservices architectures  where client processes are co-located with application server processes. Here, clients interact with respective local instances of \texttt{scarHW-server} to safely replicate requests before sending them to application processes. The separation between \ac{name} processes, assuming mostly synchronous behavior thanks to deployment on programmable hardware, and (asynchronous) application processes naturally fits modern \emph{sidecar}~\cite{sidecars,sidecarsurvey} microservice architectures.
\ac{conic-us} also allows seamless integration with popular container orchestration services such as Kubernetes~\cite{kubernetes} and docker-swarm~\cite{docker-swarm}.

%% file: sections/algorithm.tex
\section{\Ac{ckc}} \label{sec:algorithm}
\acused{ckc}
\acused{ckcuc}

\subsection{Collaborative consensus}
Our novel \ac{ckc} algorithm solves (uniform) \acl{cc} -- or \acs{cc} for short --  a generalized variant of the classical (uniform) ``selective'' consensus with the following properties:

\begin{description}[font=\textnormal]
    \item[\prop{Validity}:] If a process decides set $V$, then every element $v\in V$ was \prim{propose}d by some process.
    \item[\prop{Termination}:] Every correct process eventually decides some set of values $V$.
    \item[\prop{Uniform agreement}:] No two processes $p_i, $ $p_j$ decide different sets of values $V_i \neq V_j$.
\end{description}

Thus, \ac{cc} can decide several proposed values simultaneously. 
We assume that proposed values are distinct, e.g., messages with unique identifiers.
We use \prim{crash} to denote when a process recognizes that it is faulty (delayed) and quits execution and \prim{decide} to denote a decision after which the process quits the consensus instance.
We further assume that all processes propose within a fixed time window $\Delta_{W}$.

\subsection{\Ac{ckc} algorithm}

The core logic of \ac{ckc} was inspired by the consensus algorithm for omission failures of Parv{\'e}dy \& Raynal~\cite{omission_sync_uniform_consensus:SPAA04} (similarly to Paxos for countless  asynchronous variants~e.g., \cite{Dang2015NetPaxos,SpecPaxos_PortsNSDI2015,Li2016NOPaxos,wpaxos_ailijiang_murat_2020,Derecho_Jha-Birman2019,waverunner}), with several key modifications to fit practical applications and enhance performance. 

\paragraph{Overview.}
We first give an intuitive overview of \ac{ckc} (\autoref{alg:ckcuc}, example executions in \autoref{fig:ckc_execution}) behavior and correctness. Please refer to
\iftoggle{withAppendix}{
        \autoref{apx:algorithms}
}{
        supplementary material
}
for a detailed specification.

\begin{algorithm}[t]
\caption{\ac{ckcuc}. Executed by $p_i \in \Pi$.}
\label{alg:ckcuc}
\SetCommentSty{commentfont}
\SetKw{Elif}{elif}{}
\SetKwBlock{ElifBlock}{elif}{}
\SetKw{Then}{then}{}
\BlankLine
\lnl{line:ckc_init}$V_i \gets \{proposal\}$ (can be a heartbeat)\; 
\nl$suspect_i, ~locked_i \gets \emptyset$\;
\nl\Upon(start of round $r_i$){
    \lnl{line:ckc_send}\algprim{send}{(r_i, ~V_i, ~locked_i)} \algoboxAE{to every $p_j \notin suspect_i$}\;
}
\nl\Upon(end of round $r_i$) {
    \algoboxAE{
    \begin{minipage}{0.75\linewidth}
    \lnl{line:ckc_suspectUpdate}$suspect_i \gets$ processes which have not sent $p_i$ a message in $r_i-1$\;
    \end{minipage}
    }
    \algoboxDQ{
    \begin{minipage}{0.75\linewidth}
        \lnl{line:ckc_crash}\lIf{$|suspect_i| > t$} { \prim{crash}}                
    \end{minipage}
    }
    \algoboxLK{
        \begin{minipage}{0.75\linewidth}
        \lnl{line:ckc_selfLockedCheck}\If{$p_i \notin locked_i$}{
            \lnl{line:ckc_lockedCheck}\If{$r_i > |suspect_i|$ \textnormal{\bf or} $locked_i \ne \emptyset$}{
                \lnl{line:ckc_lockedSelfAdd}$locked_i \gets locked_i\cup \{p_i\}$\;
            }  
        }
        \end{minipage}%
    }

    \algoboxDQ{%
        \begin{minipage}{0.75\linewidth}
        \lnl{line:ckc_decideLockedMajority}\Elif $|locked_i| > t$ \Then \algprim{decide}{(V_i)}\;
        \lnl{line:ckc_decideBound}\lIf{$r_i > t$}{ \algprim{decide}{(V_i)}}
        \end{minipage}%
    }
}

\lnl{line:ckc_recv}\Upon(\algprim{recv}{(r_j, ~V_j, ~locked_j)}: 
) {
    \lnl{line:ckc_roundMismatch2}\lIf{$r_j = r_i+1$}{
        end $r_i$, start $r_j$, continue below}
    \lnl{line:ckc_syncEnforce}\If{\algoboxAE{$r_j = r_i$ \textnormal{\bf and}  $p_j \notin suspect_i$}}{
        \lnl{line:ckc_vrecv}$V_i \gets V_i \cup V_j$\;
        \algoboxLK{
        \lnl{line:ckc_lockedUpdate}$locked_i \gets locked_i \cup locked_j$\;
        }       
    }
}
\end{algorithm}

\begin{figure*}
    \centering
    \begin{subfigure}{0.25\linewidth}
        \centering
        \includegraphics[width=\linewidth]{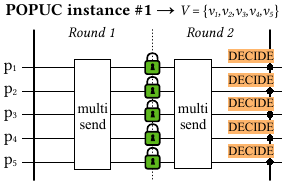}
        \caption{Failure free}
        \label{fig:ckc_execution1}
    \end{subfigure}\hfill
    \begin{subfigure}{0.375\linewidth}
        \centering
        \includegraphics[width=\linewidth]{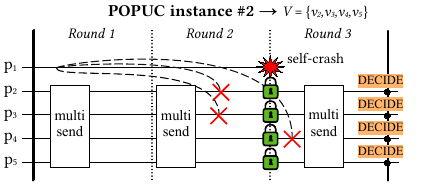}
        \caption{Late message omitted through active exclusion}
        \label{fig:ckc_execution2}
    \end{subfigure}\hfill
    \begin{subfigure}{0.375\linewidth}
        \centering
        \includegraphics[width=\linewidth]{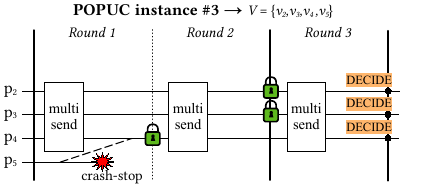}
        \caption{Crash-stop failure}
        \label{fig:ckc_execution3}
    \end{subfigure}
    \Description{Sample execution of the CKC protocol.}
    \caption{Sample \ac{ckcuc} instances (executed back-to-back in \ac{smr}). 
    Green locks represent processes in the ``locked'' state.
    }
    \label{fig:ckc_execution}
\end{figure*}

\Ac{ckcuc} uses a classical round-based approach and optimal early stopping, i.e., it is guaranteed to terminate in $\min(f+2,t+1)$ rounds\footnote{
    $min(f+2, t+1)$ is a lower bound on the number of rounds for uniform consensus in synchronous systems with omission failures~\cite{aguilera_sync_consensus_proof, sync_consensus_hybrid_faults_biely}.
}.
In every round a process $p_i$ multi-sends its set $V_i$ of proposed values it has seen so far. At the start, $V_i$ contains only $p_i$ proposal (\autoref{line:ckc_init}). When $p_i$ receives other processes' proposal sets (during the same or later rounds), it adds them to $V_i$ (\autoref{line:ckc_vrecv}). In case $p_i$ receives a message from a process in the next round, it immediately starts the next one (\autoref{line:ckc_roundMismatch2}) so that the receive routine always follow the send routine of the same round enforcing (not assuming) round alignment. The protocol uses the following mechanisms to determine whether it is safe to decide on $V_i$:

\begin{description}[font=\textnormal,leftmargin=\parindent]
    \item[\algoboxAE{Active exclusion:}] $suspect_i$ tracks possibly failed senders; these are permanently excluded from send (\autoref{line:ckc_send}) and receive (\autoref{line:ckc_syncEnforce}) operations. Late messages are also filtered out, resulting in  omission failures (cf. \autoref{fig:ckc_execution2}).
    \item[\algoboxLK{Locking:}] we say that $p_i$ is ``locked'' when it heard all the values that can be known at $r$ (and consequently no more can be learned in the future~\cite{omission_sync_uniform_consensus:SPAA04,garg2002elements}). This occurs when $p_i$ has completed more rounds than the number of processes which $p_i$ suspects are faulty (\autoref{line:ckc_lockedCheck}). Intuitively, at that point some round must have passed in which $p_i$ detected no new failure, so it has heard everything its peers knew, and by active exclusion it cannot hear from any more processes. $p_i$ can also become locked if it receives a message (with $V_j$) from another process $p_j$ that is already locked (\autoref{line:ckc_lockedUpdate} then \autoref{line:ckc_lockedCheck}). 
    \item[\algoboxDQ{Decide or quit:}] if $p_i$ suspects more than $t$ processes, it learns that it does not belong to a correct quorum and quits execution (\autoref{line:ckc_crash}).
    When $p_i$ knows of a majority ($t+1$ out of $N = 2t+1$) of locked processes, or it reaches its last round, it decides (\autoref{line:ckc_decideLockedMajority}). Requiring a genuine majority is what makes the agreement uniform: any two processes deciding this way have seen locked sets that overlap, so neither can decide on knowledge the other lacks. Processes also decide if they reach round $t+1$ (\autoref{line:ckc_decideBound}); by then, they know that at least one failure-free round has passed in which all processes fully exchanged their $V$ (i.e., the classical ``pidgeonhole'' principle~\cite{LowerboundConsistency_FischerLynchIPL1982, aguilera_sync_consensus_proof}).
    Most of the time no failures occur and the decision takes two rounds (cf. \autoref{fig:ckc_execution1}): round 1 for locking, round 2 for locked majority.
    
\end{description}

\paragraph{Important differences to Parv\'{e}dy~\&~Raynal~\cite{omission_sync_uniform_consensus:SPAA04}.}

To fit practical applications, \ac{ckcuc} differs in three key aspects.

\begin{enumerate}[leftmargin=\parindent,itemindent=1mm]

\item\textbf{Deciding on sets.}
\ac{ckcuc} 
decides a set of values $V$ rather than a single value -- solving \ac{cc} as opposed to classical ``selective'' consensus. This simple optimization is highly effective as it substantially increases 
throughput and mitigates the cost of all-to-all message relays. 
\item\textbf{Handling round mismatch and late messages.}
\ac{ckcuc} explicitly considers initial synchronization window and bounded clock drift, rather than a perfectly synchronized round counter~\cite{omission_sync_uniform_consensus:SPAA04} which makes algorithms more compact but is unimplementable in practical systems.
As a result, a process may lag behind, e.g., receive a message before even starting.
\ac{ckcuc} processes explicitly handle messages from processes one round ahead of them by immediately triggering the start of the next round in advance (\autoref{line:ckc_roundMismatch2}). Processes also explicitly discard late messages from slow peers (\autoref{line:ckc_syncEnforce}) to avoid polluting their $V$ and $locked$ sets, which would otherwise lead to inconsistent decisions (\autoref{fig:ckc_execution2}).
\item\textbf{Full value set relay to simplify verification.}
Processes in \ac{ckc} always multi-send the full value set $V_i$ instead of just the values learned in the previous round~\cite{omission_sync_uniform_consensus:SPAA04}.
This does not improve efficiency, but it makes the algorithm easier to understand and considerably cheaper to model check, which is what allowed us to verify \ac{ckcuc} with the \apalache{} model checker~\cite{OtoniFKKMOPTK:2026}.
Full-relay keeps the reachable state graph smaller than $new$-relay does, mainly by removing the per-process $new$ variable and the round-boundary action resetting it, as a message payload is simply a projection of the sender's current state.
Furthermore, \ac{ckc} increases the bandwidth consumption \textit{only in the occurrence of failures}, i.e., with more than 2 rounds, as both versions have processes multi-send $V$ in the second round. As failures are extremely rare, \ac{ckc} benefits from verified correctness at little practical cost. 

\end{enumerate}

\begin{table*}[t!]
{\small
    
    \setlength{\tabcolsep}{1.8mm}
    \caption{\acs{ckcuc} vs leader-based, multileader and leaderless asynchronous algorithms. All tolerate $t$ failures (cf.  \autoref{tab:comp}). $S$ is the number of shards. *Importantly: \ac{name}'s maximum latency is much lower than other's average latency, i.e., $\Delta_{max} \ll$ $\Delta_{avg}$. }
    \label{tab:algos}
    
 \begin{tabular}{l|ccccc}
        \toprule 

\bf  
System & \bf \makecell{\# decisions (latency) \\ \textit{without batching}} &  \bf \makecell{\# messages \\ per decision}  & \bf \makecell{Load balance} & \bf \makecell{Downtime on} & \bf \makecell{Primary bottleneck} \\
\midrule
\bf 
VR~\cite{ViewstampedReplication_OkiPODC88}, Raft~\cite{raft} 
&  1 ($2\Delta_{avg}$) & $2N$ & no & Leader crash & Leader (\ref{itm:bottleneck}) \\
\bf NOPaxos~\cite{Li2016NOPaxos} & 1 ($2\Delta_{avg}$ from client) & $2N$ & no & Leader/switch crash & Replica gets all client requests \\
\bf WPaxos~\cite{wpaxos_ailijiang_murat_2020} & 1 (2 to 4$\Delta_{avg}$) & $2N$ to $4N$ &  $S$ leaders & Leader crash & Cross-shard requests\\
\bf Derecho~\cite{Derecho_Jha-Birman2019} & 1 ($2\Delta_{avg}$) & $N + N^2$ & $N$ replicas & Replica crash & Slow replica slows others (\ref{itm:downtime})\\
\bf POPUC & $N$ $(2\Delta_{max}*)$ & $2N$ & $N$ replicas & None & Relies on $\Delta_{max}$*  \\
        
    \bottomrule
    \end{tabular}
}
\end{table*}

\paragraph{Liveness, safety and automatic model checking.}
\ac{ckc} ensures \prop{Termination} if a quorum of correct processes exists (cf. \autoref{sec:model}). Intuitively, in the worst-case $f=t$ with a different process failing at every round, processes in $Q$ still decide in $t+1$ rounds (\autoref{line:ckc_decideBound}). 
\ac{ckc} always ensures \prop{Uniform agreement} through active exclusion and locked majorities. In short, once a process $p_i$ is locked, it will never receive another new proposal since late messages from faulty processes are discarded. However, it still cannot decide since it does not know whether other correct processes will adopt the same decision. This condition, i.e., uniformity, is met once $p_i$ learns of a majority of $t+1$ locked processes in \autoref{line:ckc_decideLockedMajority}.
For pen-and-paper proofs of the algorithm's properties, we refer the reader to 
\iftoggle{withAppendix}{
\autoref{apx:algorithms}.
}{
supplementary material.
} 
To enhance correctness guarantees including for optimizations, we specified \ac{ckcuc} in TLA$^+$~\cite{Lamport:2002} and checked its properties using the \apalache{} model checker~\cite{OtoniFKKMOPTK:2026}. All properties of \ac{ckcuc} were shown to hold.
The full formal specification for independent verification are readily available upon request.\dvrnote{add them to appendix in submission} 

\subsection{Consistent replication with \ac{ckcuc}}

\paragraph{Execution.}
\ac{ckc} instances, like the ones in \autoref{fig:ckc_execution}, are executed back-to-back to achieve \ac{smr}. At a high-level, every process receives requests from clients, buffers them and proposes them at the beginning of the first incoming \ac{ckc} instance. At the end an instance, correct processes deliver values in the decision set $V$ in a deterministic order, e.g., by proposer \textsc{id}. 
In instance \#1 (\autoref{fig:ckc_execution1}), all processes successfully multi-send their proposals in round 1, becoming locked. All processes learn about a majority of locked processes in round 2, hence they all decide the entire proposal set. 
Instance \#2 (\autoref{fig:ckc_execution2}) presents a scenario where $p_1$'s messages are delayed, e.g., due to a rare network or computational delay. Processes $p_{2...5}$ exclude $p_1$ from their sender list at the beginning of round 2, leading to $p_1$ quitting execution when it realizes a majority of processes has excluded it.
In instance \#3 (\autoref{fig:ckc_execution3}), $p_5$ sends a message to $p_4$ then crashes. $p_4$ becomes locked at the end of round 2 but not $p_{2}, p_3$ which witnessed the failure. The decision is shifted to round 3 where all $p_{2..4}$ hear from a majority of locked processes and all safely \prim{decide} on the entire value set $V_{2..5}$. Importantly, \ac{name} can be set so that a process that self-quits (e.g., $p_1$ in \autoref{fig:ckc_execution2}) is not (immediately) excluded from future instances, allowing it asynchronously to ask for missing values (i.e., state transfer\dvrnote{@Patrick - should we cite something relevant?}) and start reproposing once it has caught up. 

\paragraph{Heartbeats and dynamic reconfiguration.}
Processes running \ac{smr} might send empty \emph{heartbeat} values in absence of pending proposals to avoid being excluded from the set of correct processes (cf. \autoref{sec:architecture}). Processes can choose to not send a heartbeat in order to block on purpose to allow re-configuration when the load on the system is low, e.g., scale up/down, fine-tune the round time to avoid excessive heartbeat traffic.
Custom dynamic heartbeat settings allow \ac{name} to operate in intermittent ``on and off'' mode, naturally enabling the combination with flow control and/or traffic engineering.
\iftoggle{withAppendix}{
\autoref{apx:algorithms}
}{
The supplementary material
} 
provides a detailed specification of the \ac{tob} primitive on top of \ac{ckcuc} used to achieve \ac{smr}.

    
    

        


\section{Traditional solutions vs \ac{name}}\label{sec:comparison}

\subsection{POPUC-based \ac{smr} vs asynchronous algorithms}\label{sec:algo_comparison}

\autoref{tab:algos} shows characteristics of \ac{ckcuc} when used for \ac{smr} comparing it to \ac{vr}~\cite{ViewstampedReplication_OkiPODC88} (a widely adopted Paxos-like~\cite{paxos} \ac{smr} protocol), Raft~\cite{raft} and fast state-of-the-art multi-leader/leaderless algorithms. 

\paragraph{Leveraging costless synchrony.}
\ac{ckc} shows the same \emph{best-case} message delays, complexity and fault-tolerance as Paxos-like solutions, but unlike them, provides great scalability through load-balancing and does not compromise downtime. \ac{ckc}  
to that end relies on maximum latency ($\Delta_{max} \approx$ round time) as opposed to average latency ($\Delta_{avg}$), which is however 
not an issue here: 
\ac{name}'s stability ensures that the difference between \ac{name}'s $\Delta_{avg}$ and $\Delta_{max}$ is very small (<0.5$\mu$s) and, importantly, \ac{name}'s $\Delta_{max}$ smaller than $\Delta_{avg}$ of software-based systems (2.2$\mu$s vs 3--10 $\mu$s, cf. \autoref{fig:stable_latency}). This nullifies a commonly perceived implicit cost of synchrony while preserving its benefits.
In NOPaxos, client send requests directly to replicas. \ac{conic-us} (\autoref{fig:deployment_modes}) achieves a similar effect by deployment, as we will show in the evaluation.  

\paragraph{High throughput without batching.}
\ac{ckc} and other leaderless algorithms in \autoref{tab:algos} are designed to overcome the leader bottleneck (\ref{itm:bottleneck}) of classical solutions (\ac{vr}, Raft) in order to achieve high replication throughput. 
Popular network-supported solutions NOPaxos~\cite{Li2016NOPaxos} and  
works~\cite{SpecPaxos_PortsNSDI2015, Hydra_ChoiNSDI23} use in-network serialization to avoid the usual replica ACK message exchange (i.e., phase 2B) to greatly reduce the load on the leader which is left to process client requests only. The ACK phase is still present in \ac{ckc} but is offloaded to fast \ac{snic} hardware in \ac{name} consensus engine, keeping it outside the critical path. Crucially, unlike NOPaxos,   
\ac{ckc},
WPaxos~\cite{wpaxos_ailijiang_murat_2020} and Derecho~\cite{Derecho_Jha-Birman2019} allow load distribution of client requests across a number of replicas, which we show to be a key enabler of high-throughput in the evaluation. 
Like all considered approaches, \ac{ckcuc} \emph{can also batch} requests into a single instance to further increase throughput. Importantly, the higher number of decisions reported in \autoref{tab:algos} is due to \ac{cc} and not batching, i.e., processes proposing simultaneously without impacting latency.


\paragraph{Downtime and up-to-date replicas.}
Paxos-like solutions (\ac{vr}, Raft, NOPaxos and WPaxos) require leader election upon the crash of a leader, leading to downtime. Derecho, which separates broadcast from view change, is more fragile as the failure or slowdown of any process slows or blocks other replicas (cf. \autoref{tab:algos}). 
Thanks to its different approach, \ac{ckcuc} achieves negligible downtime up to $t$ (crash or omission) failures. This is because \ac{ckcuc} tackles failures with only short, bounded extra rounds (of a few $\mu$s in \ac{name}) in which processes converge to a decision and, unlike asynchronous solutions, does not require external election/view change routines based on large unsafe timeouts which can still lead to false positives.
Furthermore, 
all \ac{ckc} replicas have an up-to-date state which clients can read, unlike leader-based solutions in which non-leader processes might have a stale state and require quorum reads leading to performance degradation (cf. \autoref{sec:eval}). 
Additionally, recovery mechanisms can exploit ``self faults'' as explicit notification of a communication issue, e.g., a message drop caused by port flapping.

\subsection{Practical advantages of \ac{name} in the absence of failures}
\label{sec:perf_advantages}
 
\autoref{fig:perf_advantages} shows a breakdown of \ref{itm:bottleneck} which \ac{name} addresses via a combination of its two performance-enhancing features: hardware acceleration and \ac{ckc}-enabled load distribution.

\begin{figure}
    \centering
    \includegraphics[width=0.99\linewidth]{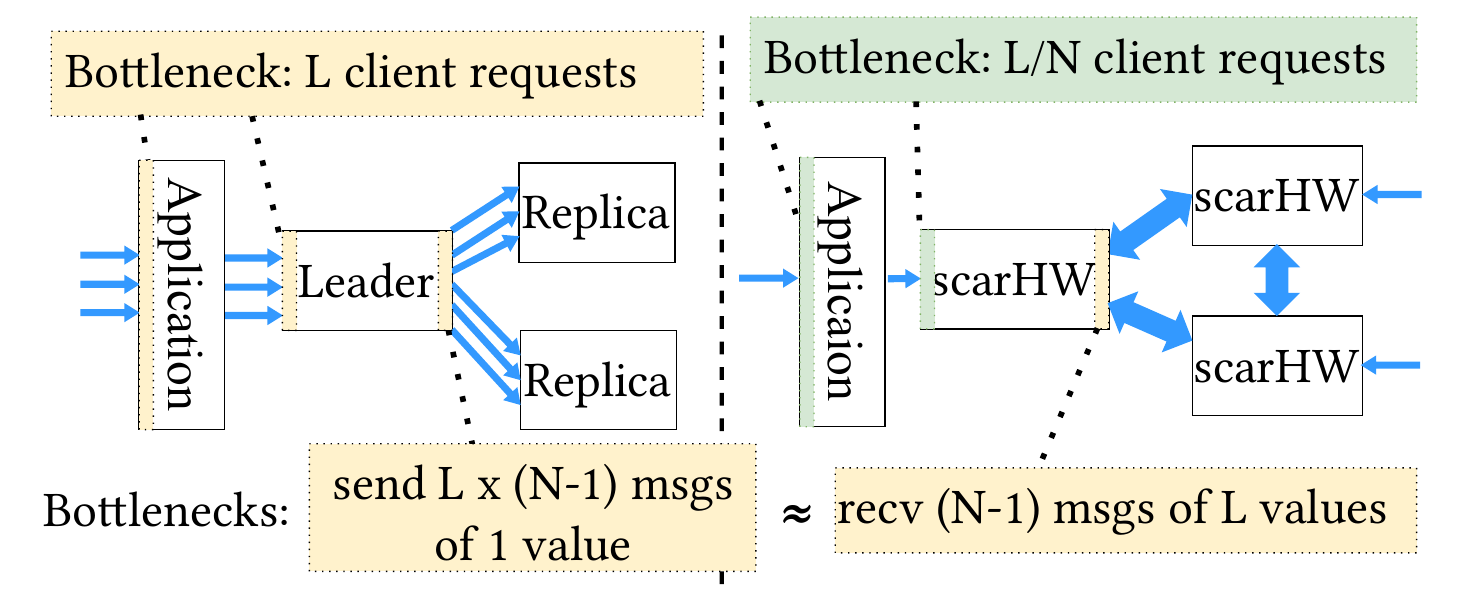}
    \caption{Bottlenecks in traditional leader-based solutions (right) vs \ac{name} (left)}
    \label{fig:perf_advantages}
\end{figure}

\paragraph{Accelerating the protocol core with hardware.}
Given a request load of $L$ client requests and a cluster of $N$ nodes, every replica in \ac{name} receives and processes $(N-1)$ \ac{ckc}'s round 2 messages, each of which contains $L$ values simultaneously proposed.
\ac{name} eliminates this major bottleneck by offloading the protocol core to \ac{fpga}-\ac{snic}, greatly outperforming state-of-the-art software-based replication solutions (cf. \autoref{sec:eval_peakperf}).

In general, under a purely algorithmic perspective, the amount of work is the same for one \ac{ckc} replica as for the leader in a leader-based \ac{smr}, i.e., they both process $L\times(N-1)$ values (cf. \autoref{fig:perf_advantages}). Practical advantages of \ac{name} come down to whether it can fit $L$ values in a single round 2 packet (remember that \ac{name}-pipeline can handle multiple \ac{ckc} instances at the same time, c.f. \autoref{fig:architecture_diagram_pipelined}), hence reducing the wasted bandwidth, e.g., packet headers. We show that this is the case for small to medium value sizes in \autoref{sec:eval_peakperf}. As outlined earlier, this process is not equivalent to batching as \ac{name} instances do not wait for $L$ values, but rather propose them as they arrive.

\paragraph{Load distribution improves application performance}
Alas, hardware-accelerated replication does not necessarily translate into performance improvements for end applications running in the endhost \ac{os}. These commonly struggle to keep up with very high throughputs, unless they are configured to handle line-rate requests, which is hard to achieve and requires a large number of resources hence likely disruptive to other applications.
This is where the advantage of \ac{ckc}-enabled load distribution comes into play: as the bottleneck shifts to the application layer (cf. \autoref{fig:perf_advantages}), each application associated with a \ac{name} replica processes only $L/N$ client incoming requests, leave the heavy lifting to \ac{name} and executes the compact decision set.
Typical leader-based replication systems, whether hardware~\cite{waverunner,Li2016NOPaxos} or software~\cite{mu,p4ce_dulong_icdcs24}, suffer from one application having to process all the client requests (\ref{itm:bottleneck}), resulting in limited effectiveness in practice, as we show in \autoref{sec:eval_apps}.

The combination of hardware acceleration and load distribution allows \ac{name} to achieve high throughput and low latency, turning the cluster size $N$ into a performance enhancer rather than a bottleneck.

%% file: sections/implementation.tex
\section{Implementation}\label{sec:implementation}

\subsection{\ac{fpga} image}
We implemented \ac{name} vanilla on an AMD/Xilinx Alveo U50 \cite{AMD_Alveo_U50} \ac{fpga} \ac{snic}. We used Vivado and the v++ compiler from the Vitis 2023.2 suite to compile and package all binary \ac{fpga} images. 
For the MAC and UDP/IP4 layer, we
used the following open source modules: UltraScale+ Devices Integrated 100G Ethernet Subsystem by AMD/Xilinx \cite{ultrascale_mac} and the project 100G-fpga-network-stack-core \cite{100Ggithub, sutter2018fpga, ruiz2019tcp}.

All other hardware modules are custom-built and use a modular design consisting of a kernel providing raw functionality and a wrapper exposing the required ports through \ac{axi}-lite interfaces to facilitate simulation and integration into the whole \ac{fpga} image. 
We employed four additional hardware timestamping modules for clock-accurate latency measurements: two between host-FPGA buffers and the consensus engine and two between the UDP/IP4 and MAC layer in both directions. Timestamps are inserted 
directly into the UDP payload for post-processing. 
All modules interact via \ac{axi} interfaces, specifically \ac{axi}4, \ac{axi}4-lite and \ac{axi}-stream. During development, we extensively evaluated latency and throughput of \ac{axi} interfaces and host-\ac{fpga} interactions, observing 
that \ac{mm2s} buffers of over 1MiB can easily saturate a 100Gbps link.
Similar iterative testing showed that kernel processing time is very stable and negligible compared to network latency. Overall, custom hardware modules amount to 1577 lines of SystemVerilog. 
Refer to anonymous \cite{anonymized_ref} for a detailed overview of the hardware implementation.

\subsection{Optimizations and software layer}
We implement a dynamic round timer adjustment mechanism in-hardware to keep replicas synchronized in the presence of clock drift or of a large initial synchonization window. In short, in the occurrence of a round mismatch (cf. \autoref{alg:ckcuc} \autoref{line:ckc_roundMismatch2}), the consensus engine adjusts the round timer of \autoref{fig:architecture_diagram_vanilla} by not only anticipating the next round but also slighty reducing the round duration. 
In our setup, \ac{name} cluster size ($N$) is limited by Ethernet's \ac{mtu} as follows: $N = (MTU - header - 17B)/(8B + 2B) \approx 140$, which is over 20$\times$ the size of common replication clusters~\cite{zookeeper, ukharon, mu, waverunner}.
\texttt{scarHW-server} is a lightweight multi-threaded TCP server implemented in 466 lines of C++ exposing an API for \ac{name} settings and (consistent) GET/SET requests.

The number of parallel consensus (green) modules for \ac{name} pipelined is limited to 31 by the available \ac{fpga} area on the Alveo U50.
\texttt{scarHW-server} interacts with \ac{name} via the AMD/Xilinx XRT library.
\ac{name} source code is available online\footnote{All source code will be made available upon publication.}.

%% file: sections/evaluation.tex
\section{Evaluation}\label{sec:eval}

We evaluate \acs{name} by comparison with state-of-the-art services and applications, addressing three research questions: 

\begin{enumerate}[ref={\bf RQ\arabic*},label={{\bf RQ\arabic*}:},leftmargin=\parindent,itemindent=6.5mm]

    \item What is \ac{name} peak performance (\ref{itm:bottleneck}, \ref{itm:ft})?\label{rq:peak}
    \item How does \ac{name} perform when combined with real-world \ac{smr} applications
    (\ref{itm:bottleneck}, \ref{itm:ft})?\label{rq:apps}
    \item  What is the impact of failures on the performance of \ac{name}-based services (\ref{itm:ft}, \ref{itm:downtime})?\label{rq:ft}

\end{enumerate}


\subsection{Methodology}\label{sec:eval_setup}

\paragraph{Evaluation cluster.}
We use a cluster 
with 3$\times$ Supermicro SYS-120U-TNR servers with Dual Xeon Gold 5315Y processors and 190GiB DDR4 RAM in line with state-of-the-art works compared against. Each server is equipped with one AMD Xilinx Alveo U50~\cite{AMD_Alveo_U50} running \ac{name} (serving up to 59 applications) and one NVIDIA Mellanox ConnectX-7 used by all other compared approaches. Both \acp{nic} are connected to an Edgecore Wedge100BF-32X-O-AC-F Tofino switch via a 100G DAC QSFP28 cable.
All endhosts run Ubuntu Server 22.04 with Linux 5.15.0-125 and use \ac{xrt} version 2.16.204, \texttt{xocl} and \texttt{xclmgmt} from Vitis 2023.2.

\paragraph{Comparison.}
We compare \ac{name} against the following state-of-the-art replication systems:

\begin{description}[font=\it]
\item[Waverunner\textnormal{~\cite{waverunner},}] the fastest \ac{smr} module to our knowledge, 
using Alveo U280~\cite{alveo280} \acsp{fnic} to accelerate failure-free execution of Raft~\cite{raft}.

\item[Mu\textnormal{~\cite{mu},}] a low-latency Paxos-based algorithm leveraging one-sided \ac{rdma} writes.

\item[P4ce~\textnormal{\cite{p4ce_dulong_icdcs24},}] a protocol similar to Mu implementing \ac{rdma} write multicast on programmable switches to reduce network load on the leader.

\item[NOPaxos\textnormal{~\cite{Li2016NOPaxos},}] a fast \ac{smr} protocol offloading request serialization to a network switch to eliminate \ref{itm:bottleneck}. 

\end{description}

We thoroughly evaluate the integration of \ac{name} into two widely used applications -- Redis~\cite{redis}, a (non-consistent) distributed key-value store, and Zookeeper~\cite{zookeeper}, a consistent distributed configuration service -- by comparing their performance with and without \ac{name}. Additionally, for Redis, we compare \ac{name} against RedisRaft~\cite{redisraft}, an official module developed by RedisLabs that implements Raft-based 
\ac{smr}.

\paragraph{Breakdown.}
The first set of microbenchmarks (\autoref{sec:eval_peakperf}) puts \ac{name} in perspective w.r.t. the state of the art in terms of peak performance in failure-free executions. These isolate the effect of hardware acceleration (cf. \autoref{sec:perf_advantages}) showing the advantages of \ac{name} and Waverunner.
Full-system benchmarks (\autoref{sec:eval_apps}) highlight the benefits of \ac{ckc}-enabled load distribution combined with hardware acceleration when \ac{name} is integrated into real-world applications, i.e., the key performance features of \ac{name} (cf. \autoref{sec:perf_advantages}).
Finally, we evaluate the impact of failures on \ac{name}-based services (\autoref{sec:eval_ft}).

\begin{figure}
    \centering
    \includegraphics[width=\linewidth]{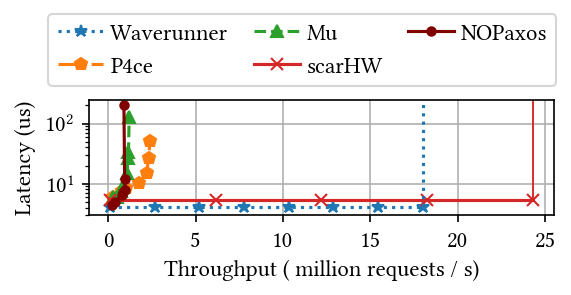}

    \caption{Latency at increasing throughput (64B requests, 5 replicas). Vertical lines show hardware bottlenecks.}
    \label{fig:max_perf}
\end{figure}

\begin{figure}
    \centering
    \includegraphics[width=\linewidth]{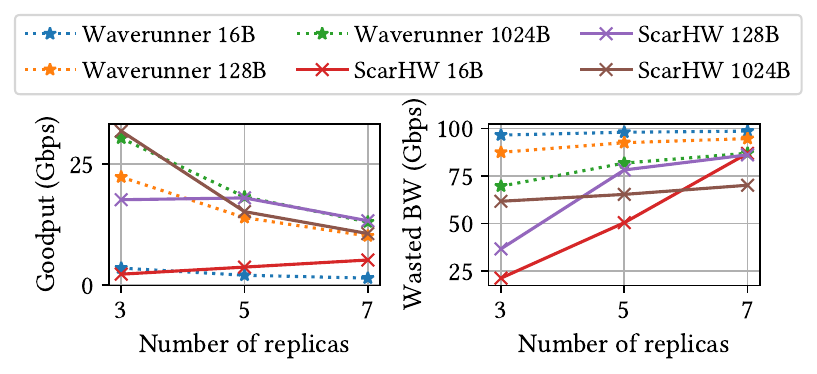}
    
    \caption{Goodput and relative wasted bandwidth (e.g., headers, \textsc{ack}s) of hardware replication systems with increasing cluster and proposal value sizes.
    }
    \label{fig:tput_scale}
\end{figure}

\subsection{Microbenchmarks: peak performance (\ref{rq:peak})}\label{sec:eval_peakperf}

\paragraph{Settings.}

We re-implement the NOPaxos ``normal operation'' protocol~\cite{Li2016NOPaxos} assuming perfectly ordered, lossless packets, hence capturing its best-case performance (and outperforming the full protocol’s latest evaluation~\cite{Hydra_ChoiNSDI23}).
Due to the impossibility of fully re-implementing other compared approaches, all requiring very specialized hardware (namely specific \ac{snic} models, Infiniband, P4-enabled switch), we simulate their performance based on their complexities and using real measurements to infer parameters.
To obtain the best-case consensus decision latency for prior methods, we measure 1 million \ac{fpga}-\ac{fpga} packet \acp{rtt} for Waverunner and single \ac{rdma} write latencies for Mu and P4ce at a low throughput of 1K packets-per-second with payloads from 64B to 1024B. 
We use the minimum latency observed to scale throughput curves taken from~\cite{waverunner} for Waverunner and~\cite{p4ce_dulong_icdcs24} for Mu and P4ce. We use the same benchmark settings as both cited works (with similar server specifications) to approximate real performance as closely as possible. 
We measure latency and throughput of \ac{name} vanilla by running several billion consensus instances with a round period of 2.7$\mu$s. This includes a safety margin of 0.3$\mu$s over the worst-case latency ever observed although not required for correctness or performance. 
Values are proposed by a leader node for all approaches except \ac{name}, in which all replicas propose values simultaneously and NOPaxos in which clients broadcast requests to all replicas. 

\paragraph{Latency vs throughput.}
\autoref{fig:max_perf} shows the replication latency at increasing throughput with a standard cluster size of 5 replicas, using small 64B packets. 
Performance differences across the evaluated approaches are primarily driven by whether they rely on software for handling replication requests or not.
NOPaxos, Mu and P4ce are constrained by the host's software stack, hitting a  maximum throughput wall at around $\sim$1Mpps. 
Among these, P4ce achieves the best performance by leveraging \ac{rdma} and switch-level message aggregation, which reduces load on the leader (\ref{itm:bottleneck}).
In contrast, the \ac{fpga}-offloaded designs Waverunner and \ac{name} operate at line rate on dedicated hardware, maintaining ultra-low, stable latency up to extremely high throughput. Both approaches saturate the 100Gbps network: Waverunner's large number of pipelined proposals clog the output bandwidth of a single \ac{snic}, while \ac{name}'s \ac{ckc} round 2 messages saturate the replicas ingress badwidth. \ac{name} advantage comes from condensing multiple values into a single packet, which reduces the impact of packet headers to the overall throughput (see \autoref{fig:perf_advantages}).
Waverunner shows slightly better latency than \ac{name} (5.4$\mu$s vs 3.41$\mu$s) because the latter is limited by $\Delta_{max}$ (cf. \autoref{tab:algos}) which also includes a safety margin.

\paragraph{Throughput at scale and batching.}
\autoref{fig:tput_scale} shows goodput and wasted bandwidth (throughput - goodput) per replica of \ac{name} and Waverunner with standard \ac{smr} cluster sizes~\cite{mu,p4ce_dulong_icdcs24,waverunner,ukharon} (cf. \ref{itm:ft}) and different proposal value sizes, to showcase the effect of batching on throughput.
While \ac{name}'s pipelined parallel instances for 16B and 128B at 3 nodes maximize the \ac{fpga} area of the Alveo U50, commonly available \ac{fpga} \acp{snic} with larger area, e.g., Alveo U280~\cite{alveo280}, would extend \ac{name}'s improvements at small scale. 
All other points for both approaches saturate a 100Gbps network.
In general, \ac{name} achieves better goodput best on small to medium value sizes and less wasted bandwidth because it can pack all values of a single consensus instance into a single packet, minimizing the impact of packet headers.  
Waverunner shows slightly better goodput (but higher wasted bandwidth) for large 1024B values because \ac{name}-pipelined is limited by the number of parallel instances (cf. \autoref{fig:architecture_diagram_pipelined}) that it can run without exceeding the bandwidth. For instance, with 5 replicase, \ac{name} would need to run 2.5 parallel instance to saturate the network, but it defaults to 2 since it needs to keep the same size for all instances.

\begin{figure*}[t]
    \centering
    \includegraphics[width=\textwidth]{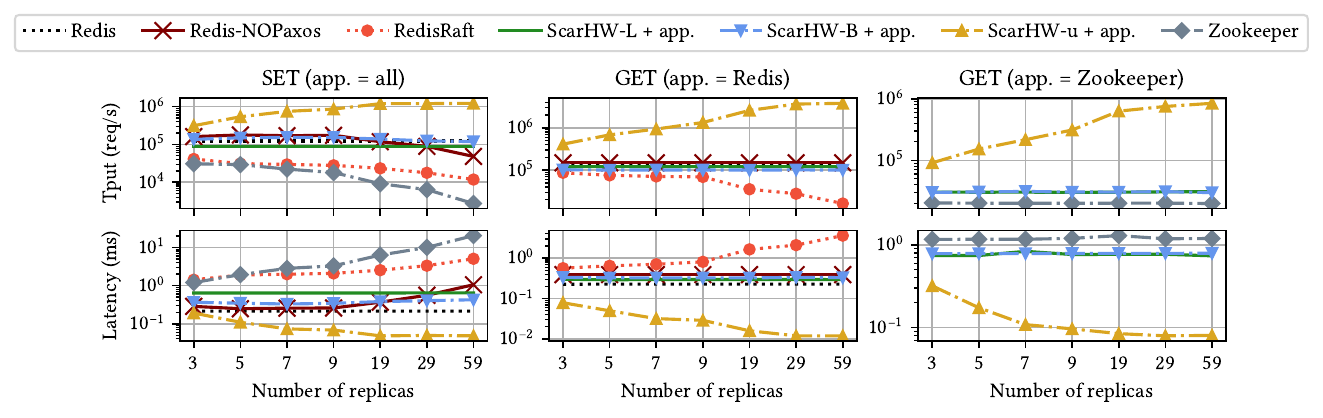}
    
    \caption{Application performance comparison with respect to the number of replicas. Throughput is on the $y$-scale in the top row (higher is better) and latency in the bottom row (lower is better); both use logarithmic $y$-scales.}
    \label{fig:apps}
\end{figure*}

%


\subsection{End-to-end application performance (\ref{rq:apps})}
\label{sec:eval_apps}

Network-speed processing rates of hardware replication systems such as \ac{name}, Waverunner, and similar systems~\cite{Nanoconsensus_RovelliSoCC25, consensus_in_a_box} put huge pressure on endhost (software) applications which can easily become the bottleneck unless they use heavily-parallelized, dedicated packet processing pipelines~\cite{waverunner, Nanoconsensus_RovelliSoCC25}.  
Crucially, \ac{name} also accelerates applications which go at ``software speed'' out-of-the box thanks to \ac{ckc}'s properties (cf. \autoref{sec:perf_advantages}), namely by allowing load distribution.

\paragraph{Settings.}
This set of benchmarks tests the failure-free performance of our \ac{name} vanilla implementation integrated into Redis and Zookeeper to provide consistent replication. 
We evaluate latency and throughput of SET and GET requests from 60 clients, each sending 100k requests for each measurement using 24 keys and random 8B integers as values. SET requests are blocking
and servers also do not use batching.
Cluster sizes larger than 3 co-locate replicas uniformly across the 3 physical nodes. 
Zookeeper SET requests, NOPaxos GET, and RedisRaft SET requests are forwarded directly to the leader as required. For all approaches we evaluate consistent GET requests: we set the quorum read option in RedisRaft~\cite{redisraft_quorum_reads} and issue GET requests to a follower node for Zookeeper, forcing it to synchronize with the leader beforehand. For \ac{name}, GET requests are consistent by default as all nodes are updated synchronously.
We also include results for un-replicated Redis as a baseline. 
\ac{name} is evaluated in all its deployment modes (cf. \mbox{\autoref{sec:design}}). In \ac{conic-lb}, we use Linux IP Virtual Server~\cite{ipvs_linux_virtual_server} as a thin client-side load balancer to forward request to all replicas using a round-robin scheme. 
SET requests for \ac{name}-based Redis and Zookeeper are handled by \texttt{scarHW-server}, which consistently replicates the request through \ac{name}, 
responds to clients, and asynchronously delivers the request to the application. Waverunner~\cite{waverunner} adopts a similar strategy. 
We use the official \texttt{redis-benchmark}~\cite{redisbenchmark} for Redis and RedisRaft.
For Zookeeper, we implement simple C clients using the Zookeeper API. We use the same implementation of NOPaxos described in \autoref{sec:eval_peakperf}, with each replica interfacing with a local Redis instance.

\paragraph{Results.}
\autoref{fig:apps} shows that \ac{name} applications \textit{maintain or even improve performance when scaling} to more replicas, as opposed to the typical performance degradation of \ac{smr} observed with RedisRaft, native Zookeeper, and leaderless NOPaxos at large scale.  
This is due to \ac{name}  benefiting from (1) synchrony, i.e., load distribution in this specific benchmark, without suffering from (2) relying on conservatively high latency bounds thanks to protocol hardware offload (cf. \autoref{sec:algo_comparison}, \autoref{sec:perf_advantages}).
For instance, with 59 replicas, every Redis/Zookeeper application instance with \ac{conic-lb} or \ac{conic-us} has to handle $\sim$1 client request at a time, with \ac{name} handling the core protocol heavy-lifting. \ac{conic-l}, RedisRaft, and Zookeeper each replica or leader has to handle 60 requests in software (\ref{itm:bottleneck}). While NOPaxos significantly reduces the load on the leader by offloading the leader to the network, replicas do not benefit from (1) as they still have to process all client requests.  Avoiding (2) makes the cost of \ac{ckc} negligible w.r.t. the speed at which application threads can receive/respond to client requests, resulting in substantial scaling improvements (\ref{itm:ft}). Overall, \ac{conic-lb} outperforms native Zookeeper and RedisRaft in all metrics and performs slightly better than Redis-NOPaxos in throughput and GET latency despite the load-balancer and protocol overhead. 
\ac{conic-us} embodies the benefits of distribution since clients send requests to respective \emph{co-located} replicas, achieving from 2$\times$ to over $100\times$ better performance than all approaches in all metrics. 
The improvements decrease 
at yet larger scales 
as application servers start idling due to the small number of client requests. 

Due to the \ac{smr} overhead, replication clusters are usually limited to a small number of replicas reaching at most 7 or 9~\cite{zookeeper,mu,electrode,ukharon},
explaining the drop in performance of RedisRaft and Zookeeper. \ac{name} overturns this common 
constraint making replication an appealing tool to distribute load and increase performance, particularly for microservices.


\subsection{Failure impact (\ref{rq:ft})}
\label{sec:eval_ft}

\paragraph{Description.}
This benchmark evaluates the impact of failures on \ac{conic-lb} and \ac{conic-us} integrated into Zookeeper by injecting process failures in a 
5-replica cluster. We omit direct evaluation of omission failures 
since \ac{ckcuc} treats crashes and omissions the same way (they have the same effect), and we actually observed no message drops in our setup with 1\% (0.2$\mu$s) safety margin over our 40-days-long latency stability benchmark (cf. \autoref{fig:stable_latency}).
Just like for the previous benchmark, \ac{name} does not use the optimization of only relaying value \textsc{id}s. 
We monitor the throughput of SET requests from 60 concurrent clients over 10s, and manually stop the leader for Zookeeper and one of the replicas for \ac{name}-Zookeeper at $\sim$4s. Clients send 100K requests each.
For Zookeeper, we configure clients to continuously query the cluster upon failure and immediately switch to the new leader once leader election terminates. In \ac{conic-lb} + Zookeeper, the load-balancer reconfigures itself to remove the faulty replica, while in \ac{conic-us} we assume that the clients associated with the faulty replica crash alongside it, thus terminating co-located microservices. 
We also report Waverunner metrics from the original evaluation~\cite{waverunner} to give a rough comparison with another state-of-the-art \ac{smr} system using asynchronous consensus (Raft).

\paragraph{Results.}
\autoref{fig:failure_impact} shows the inherent advantage of \ac{name}'s \ac{ckcuc} with respect to asynchronous consensus algorithms used by both native Zookeeper (ZAB~\cite{zab}) and Waverunner (Raft~\cite{raft}).
Since \ac{ckcuc} is fully peer-based, failures cause negligible slowdown (cf. \ref{itm:downtime}), which we experimentally find to be one extra round of \ac{ckcuc} (2.7$\mu$s) most of the time. The effect is a decrease in throughput due to the increase in request load on every active replica. The election procedure upon leader failure in both Zookeeper and Waverunner results in a much longer \textit{downtime} of hundreds of ms where clients stall until the service comes back up. When the election terminates (with no recovery), Zookeeper and Waverunner do not observe a change in throughput as the same request load appies to the new leader node, which processes the same number of acknowledgments from replicas.
Recoveries in \ac{name} are treated as process joins resulting from membership changes proposed and agreed upon through \ac{ckcuc} which trivially follow the inverse trend of failover leading to an \emph{increase} in throughput.
Failure of $t + 1$ replicas results in a system stall for \ac{name} as well as for asynchronous consensus services.

\begin{figure}
    \centering
    \includegraphics[width=\linewidth]{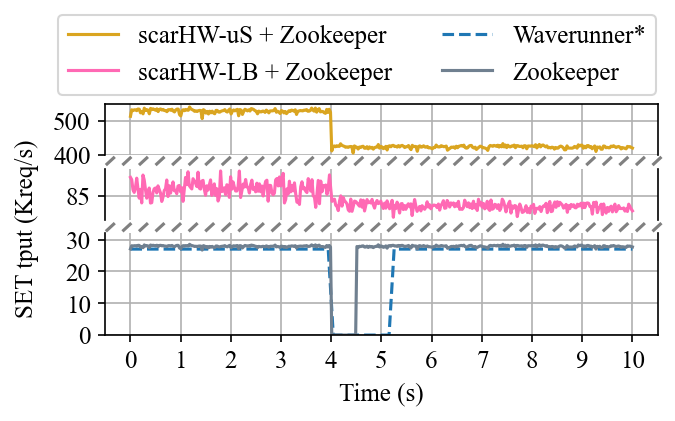}

    
    \caption{Impact of a failure on leader-based consensus applications vs applications using \ac{name}. Note that -- for fair comparison against Waverunner (*data points taken from \cite{waverunner}) -- 
    all approaches show performance with reduced number of replicas after the failure (no process  replacement
    ).}
    \label{fig:failure_impact}
\end{figure}


%% file: sections/relatedWork.tex
\section{Related Work} \label{sec:relatedWork}

\paragraph{Asynchronous algorithm acceleration.}

Mu~\cite{mu}, P4ce \cite{p4ce_dulong_icdcs24}, and $\mu$Kharon~\cite{ukharon}  propose asynchronous consensus-based algorithms heavily optimized for \ac{rdma} networks. 
Consensus in a Box~\cite{consensus_in_a_box} spearheaded \ac{fpga}-based coordination by offloading  \Ac{zab} to a \ac{snic}. 
Waverunner~\cite{waverunner} accelerates only failure-free operations of Raft with a \ac{nic}, leaving failover routines to software.
Similarly, systems to offload Paxos to network switches have been explored~\cite{paxos_in_the_nic, Dang2015NetPaxos}.
Notably, NOPaxos~\cite{Li2016NOPaxos} (follow-up of SpecPaxos~\cite{SpecPaxos_PortsNSDI2015}) is a high-performance ``leaderless'' solution that instruments network switches to achieve network ordering to effectively circumvent \autoref{itm:bottleneck}. 
\ac{name} goes beyond bare acceleration of asynchronous algorithms introducing a novel consensus algorithm which leverages synchrony to enable load-balancing, greatly improving throughput over solutions above. 

Note that many asynchronous consensus works 
assume (eventually) reliable communication~\cite{BCBT96} (via resends), while quorums (used originally to handle process failures) can inherently handle \emph{some}  message losses (
cf. \autoref{tab:comp}). The extent   
of this has only been explicitly considered   
recently~\cite{boundbft}.

\paragraph{Synchronous algorithms.}

The consensus problem in synchronous systems has been well explored in theory 
for (pure) process crash-stop failures ~\cite{distributed_algos_lynch,sync_consensus_summary_raynal,The_timely_comp_tcb_Verissimo_2002}. 
Like \ac{ckcuc}, 
many synchronous protocols 
could also reach several decisions, but choose one deterministically to implement standard (selective) consensus. \ac{ckcuc} can be trivially adapted. 

Unlike recent systems using synchrony for efficient coordination like FiDe~\cite{fide} and Nano-consensus~\cite{Nanoconsensus_RovelliSoCC25}, \ac{name} 
tolerates communication failures, increasing robustness. Chora \cite{ChoraSync_LiArxiv25} uses synchrony for pipelining, but carries on the drawbacks of the partially synchronous model.

A few works consider synchrony with omission 
~\cite{omission_sync_uniform_consensus:SPAA04, sync_consensus_summary_raynal, sync_consensus_hybrid_faults_biely, kset_agreement_sync_omission_parvedy_raynal} and timing failures~\cite{CF99}. \ac{name} is the first 
system leveraging the predictability of modern \acp{fpga} for efficient distributed coordination while tolerating omission failures.

The seminal XFT
~\cite{xft:osdi16}, which uses a synchronous quorum, 
as well as several recent works~\cite{vmware_sync_bft,LWSBFT,synchotstuff}, assume 
synchrony for \ac{bft} (even in wide-area setups~\cite{boundbft}) on top of a commodity software stack without support for it. 
\ac{name} uses dedicated resources on network hardware to support synchrony ground-up.

\paragraph{Partitioning and relaxing consistency.}

Several multi-leader and leaderless variants of Paxos~\cite{Derecho_Jha-Birman2019,EgalitarianPaxos_Moraru2013,wpaxos_ailijiang_murat_2020,paxos_variants_perf_murat_SIGMOD19,VPaxos_LamportMZ09,Marandi2010RingPaxos, Dang2015NetPaxos} 
mitigate \ref{itm:bottleneck} and \ref{itm:ft} through partitioning/sharding and separation of group management from data replication logic. Derecho~\cite{Derecho_Jha-Birman2019} 
separates asynchronous atomic broadcast and view change achieving high throughput at the expense of latency degradation with failures or stragglers (\ref{itm:downtime}).
\ac{name} fully embraces synchrony 
proposing a consensus engine suited for the stability of datacenter \acp{snic} that allows it to improve over previous systems by achieving zero downtime and low latency while being scalable also \emph{within a single partition}.
Clearly, \ac{name} can also benefit from partitioning or sharding at very large scales.

Similarly to batching and sharding, limiting consensus to a small set of $k$   ``consensus(-as-)service''~\cite{consserv} nodes can help mitigate bottlenecks (tolerating more failures among ``general'' nodes), but retains fundamental limitations of the underlying algorithm (scaling poorly to more replicas $k$) which 
be improved with our approach though. Lastly, \acl{cc} differs from $k$-set agreement~\cite{kset}, where processes each decide one from 
at most $k$ different values. 

%% file: sections/conclusion.tex
\section{Conclusions} \label{sec:conclusion}


We present \ac{name}: a network card design that provides fast and robust \ac{smr} to highly available distributed applications running in datacenters. \ac{name} differs from traditional approaches in that it uses a novel synchronous consensus algorithm which exploits the reliability and stability of modern datacenter networks. Through careful algorithm design on an off-the-shelf \ac{fpga} \ac{snic}, \ac{name} mitigates the traditional cost of synchrony while preserving its key benefit: inherently leaderless coordination. 
\ac{name}-based services can easily scale (hence supporting better fault tolerance) while also improving performance, overturning common system design principles.

%% file: sections/appendix.tex
\section{\Acf{ckc}} \label{apx:algorithms}
This section provides the fully-detailed specification of the \ac{ckc} protocol, and provides pen and paper proofs for its correctness.

\subsection{Detailed specification}
\Ac{ckcuc} (\autoref{alg:apx_ckc}, the fully-detailed version of the protocol presented in 
\iftoggle{withAppendix}{
\autoref{alg:ckcuc})
}
{
the main text
}
uses a classical round-based approach and optimal early stopping. The core logic of \ac{ckcuc} was inspired by the algorithm for omission failures of Parv{\'e}dy \& Raynal~\cite{omission_sync_uniform_consensus:SPAA04}, with substantial differences to fit practical applications, as we will discuss in \autoref{sec:apx_algo_variants}.

\paragraph{Properties of \acl{cc}}
\ac{ckcuc} solves a variant of consensus that we call \acf{cc}.
Uniform \ac{cc} has a \prim{propose} downcall and \prim{decide} upcall, satisfying the following properties:

\begin{description}[font=\textnormal]
    \item[\prop{Validity}:] If a process 
    \propprim{decide}s 
    set $V$, then every element $v$ of $V$ was \prim{propose}d.
    \item[\prop{Termination}:] Every correct process eventually \propprim{decide}s some set of values $V$.
    \item[\prop{Uniform agreement}:] No two processes $p_i, $ $p_j$ \propprim{decide}~different sets of values $V_i \neq V_j$.
\end{description}

%
Thus, \ac{cc} can decide several proposed values simultaneously. 
We assume that proposed values are distinct, e.g., messages with unique identifiers.
We use \prim{crash} to denote when a process recognizes that it is faulty and quits execution.
We also assume that processes quit execution after the \prim{decide} event and that all calls execute atomically.

\begin{algorithm}[t]
\caption{\Ac{ckc} (detailed). Executed by every $p_i \in \Pi$.}
\label{alg:apx_ckc}

\SetKw{Elif}{elif}{}
\SetKwBlock{ElifBlock}{elif}{}
\SetKwBlock{Func}{func}{}
\SetKw{Then}{then}{}

\BlankLine
\nl$V_i, ~suspect_i, ~locked_i \gets \emptyset$\;
\nl$roundLiveset \gets \Pi$\;
\nl$r_i \gets 0$\;

\BlankLine

\nl\To(\algprim{propose}{(val)}:) {
    \lnl{line:ckcapx_proposevalupdate}$V_i \gets \{val\}$\;
    \nl\prim{start-timer}
}

\nl\Upon(\prim{timeout}:) {
    \algoboxAE{
    \begin{minipage}{0.5\linewidth}
    \lnl{line:ckcapx_suspectUpdate}$suspect_i \gets \Pi \setminus roundLiveset$\;
    \end{minipage}
    }

    \algoboxDQ{
    \begin{minipage}{0.75\linewidth}
    \lnl{line:ckcapx_crash}\lIf{$|suspect_i| > t$} { \prim{crash} }
    \end{minipage}
    }

    \algoboxLK{
    \begin{minipage}{0.75\linewidth}
    \lnl{line:ckcapx_selfLockedCheck}\If{$p_i \notin locked_i$}{
        \lnl{line:ckcapx_lockedCheck}\If{$r_i > |suspect_i|$ \textnormal{\bf or} $locked_i \ne \emptyset$}{
            \lnl{line:ckcapx_lockedSelfAdd}$locked_i \gets locked_i\cup \{p_i\}$\;
        }  
    }
    \end{minipage}%
    }

    \algoboxDQ{%
    \begin{minipage}{0.75\linewidth}
    \lnl{line:ckcapx_decideCrash}\Elif $|locked_i| > t$ \Then \algprim{decide}{(V_i)}\;
    \lnl{line:ckcapx_decidebound}\lIf{$r_i > t$}{ \algprim{decide}{(V_i)}}
    \end{minipage}%
    }

    \lnl{line:ckcapx_nextRound}$r_i \gets r_i + 1$ \tcp{advance to next round}
    \lnl{line:ckcapx_livesetReset}\algoboxAE{$roundLiveset \gets \{p_i\} \cup locked_i$}\;
    \nl\prim{restart-timer}\;
    \lnl{line:ckcapx_multiSend}\algprim{send}{(r_i, ~V_i, ~locked_i)} \algoboxAE{to every $p_j \notin suspect_i$}\;
}

\lnl{line:ckcapx_recv}\Upon(\algprim{recv}{(r_j, ~V_j, ~locked_j)} 
:) {
    \lnl{line:ckcapx_roundMismatch2}\lIf{$r_j = r_i + 1$}{Lines \ref{line:ckcapx_suspectUpdate}-\ref{line:ckcapx_multiSend}}
    \lnl{line:ckcapx_syncEnforce}\If{\algoboxAE{$r_j = r_i$ \textnormal{\bf and}  $p_j \notin suspect_i$}}{
        \lnl{line:ckcapx_vrecv}$V_i \gets V_i \cup V_j$\;
        \algoboxLK{
        \lnl{line:ckcapx_lockedUpdate}$locked_i \gets locked_i \cup locked_j$\;
        }
        \lnl{line:ckcapx_livesetUpdate2}\algoboxAE{$roundLiveset \gets roundLiveset \cup \{p_j\} \cup locked_i$}
    }
}
\end{algorithm}

\paragraph{System model reminders and timer primitives.}
\ac{ckc} assumes a synchronous model with crash-stop and omission failures, like \cite{omission_sync_uniform_consensus:SPAA04}, to better capture the stability of programmable network hardware in modern datacenter environments. Synchrony assumptions are required to guarantee liveness, but not safety, as we will discuss shortly. We break down the synchrony assumptions into three parameters: $\Delta_I$, $\Delta_D$ and $\Delta_W$.
$\Delta_I$ bounds most interaction latency between any two processes $p_i, p_j \in \Pi$ , i.e., communication between processes together with end-to-end processing happens most of the time within $\Delta_I$. 
$\Delta_D$ bounds clock drifts between two processes, which is easily achieved in practice and assumed by state-of-the-art datacenter coordination works, including for correctness, e.g.~\cite{ukharon}.
We further assume that all processes call \prim{propose} within a fixed time window $\Delta_{W}$.
\ac{ckcuc} interacts with a timer abstraction through the \prim{start-timer} and \prim{restart-timer}. The timer calls \prim{timeout} every time interval of duration $\Delta_I + \Delta_W + \Delta_D$ , i.e., the maximum delay in which correct processes can expect to receive a message. 

\paragraph{Failure model reminders.}
A process is \emph{faulty} if it commits a crash-stop or at least one omission failure, otherwise it is \emph{correct}. Omission failures occur when processes fails to send of receive a message, capturing both message loss and arbitrary delays. 
We consider a system of $N = 2t + 1$ processes with at most $t$ faulty processes (crash-stop or omission failures). A quorum of correct processes $Q$ exists, where $|Q| \ge t+1$ with $t$ being the maximum number of faulty processes. We use $f$ to denote the actual number of faulty processes in a given execution. \ac{ckc} is live as long as a quorum of correct processes exists, and always safe, even if $f > t$.

\paragraph{Algorithm (\autoref{alg:apx_ckc}).}

Unlike common specifications of synchronous algorithms, rounds are driven by the timer abstraction rather than by a global round counter: \prim{propose} (called asynchronously by processes) arms the timer through \prim{start-timer}, while \prim{timeout} closes the current round (lines \ref{line:ckcapx_suspectUpdate}-\ref{line:ckcapx_nextRound}) and opens the next one (line \ref{line:ckcapx_nextRound}-\ref{line:ckcapx_multiSend}). 
The body of the \prim{timeout} routine, i.e., the round change routine, is also invoked in advance through whenever a message from a process that is one roundfurther ahead is received (\autoref{line:ckcapx_roundMismatch2}). 
This mechanism ensures that processes always multi-\prim{send} their value before \prim{recv}ing any message from that round, proceeding synchronously, unlike common synchronlous algorithms which assume perctly synchronized rounds~\cite{sync_consensus_summary_raynal,sync_consensus_hybrid_faults_biely}.
We will discuss below how, in practice, $p_i$ never receives messages with a round larger than $r_i + 1$.

At the core, in every round, a process $p_i$ multi-sends its own proposal set $V_i$ containing all the values it knows together with its $locked_i$ set (\autoref{line:ckcapx_multiSend}); in the first round, $V_i$ will contain only the $p_i$'s proposal (\autoref{line:ckcapx_proposevalupdate}). 
Incoming values and locks are merged into $V_i$ and $locked_i$ (\autoref{line:ckcapx_vrecv} and \autoref{line:ckcapx_lockedUpdate}) and relayed in later rounds.

The protocol uses the same mechanisms as 
\iftoggle{withAppendix}{
    \autoref{alg:ckcuc}
}
{   
    the main paper
} 
 (highlighted with matching colors) to determine whether it is safe to decide on $V_i$:

\begin{description}[font=\textnormal,leftmargin=\parindent]
    \item[\algoboxAE{Active exclusion:}] $roundLiveset$ records the processes from which $p_i$ has received a message in the current round (\autoref{line:ckcapx_livesetUpdate2}) and is reset to $\{p_i\}$ whenever a round is closed, so that suspicions are always recomputed on evidence gathered within a single round. At the start of a round, every process missing from $roundLiveset$ is added to $suspect_i$ and is permanently excluded from the send (\autoref{line:ckcapx_multiSend}) and receive (\autoref{line:ckcapx_recv}) operations of all future rounds: an excluded process can never re-enter $roundLiveset$, hence it remains suspected forever. This mechanism filters out late messages, so arbitrary delays are turned into omission failures. Hence, a process can never learn a new value from a process it suspects.
    Active exclusion also implies that $p_i$ can only receive messages from processes which are at most one round ahead of it, i.e., $r_j = r_i + 1$, since a process that is two rounds ahead will have already excluded $p_i$ from its send set.
    \item[\algoboxLK{Locking:}] a process $p_i$ is ``locked'' when it knows all the values that can be known in the current round, and consequently no further value can be learnt. This mechanism is presented in related synchronous protocols and it is necessary to guarantee \prop{Uniform Agreement}~\cite{omission_sync_uniform_consensus:SPAA04,garg2002elements,sync_consensus_summary_raynal}. The locking condition evaluates when $p_i$ has completed more rounds than the number of processes it suspects (\autoref{line:ckcapx_lockedCheck}), since by then it has heard everything its peers know, including (possibly) values from processes that $p_i$ has suspected earlier than its peers. Intuitively, thanks to the active exclusion mechanism, late messages from $p_j \in suspect_i$ cannot reach $p_i$ after it is locked, hence it will not learn any new values. 
    Also, $p_i$ updates its locked set when it receives a non-empty $locked_j$ from another process $p_j$, in which case $p_i$ adopts $locked_j$ (\autoref{line:ckcapx_lockedUpdate}) and later adds itself to it (\autoref{line:ckcapx_lockedSelfAdd}) since it has now all the values that $p_j$ knows.
    \item[\algoboxDQ{Decide or quit:}] if $p_i$ suspects more than $t$ processes, it learns that it is faulty, i.e., does not belong to the quorum $Q$ of correct processes. Hence, it declares itself faulty and quits execution (\autoref{line:ckcapx_crash}). 
    $p_i$ decides early when it receives a message from a majority of locked processes, i.e., $|locked_i| \ge t+1$ (\autoref{line:ckcapx_decideCrash}); $t+1$ out of $N = 2t+1$ is a majority, so any two processes that decide this way have seen locked sets that overlap, and no process can decide on knowledge that such a majority does not share. This is what guarantees the uniformity of the agreement property. 
    Even if the ``locked majority'' condition is not met, processes can still decide in in round $t+ 1$ if they are still alive by then (\autoref{line:ckcapx_decideCrash}). 
    The sufficiency of $t+1$ follows from the classical pigeonhole argument~\cite{LowerboundConsistency_FischerLynchIPL1982, aguilera_sync_consensus_proof}: faulty processes can each spoil at most one round, so among $t+1$ rounds at least one must be failure-free (pigeonhole); once that clean round happens, all correct processes end up with the same information. 
\end{description}

The protocol terminates in $\min(f+2, t+1)$ rounds, a proven lower bound -- a proven lower bound for uniform consensus in synchronous systems with omission failures~\cite{aguilera_sync_consensus_proof, sync_consensus_hybrid_faults_biely}.

\paragraph{Important differences to Parv\'{e}dy~\&~Raynal~\cite{omission_sync_uniform_consensus:SPAA04}.}

To fit practical applications, \ac{ckcuc} differs in three key aspects.

\begin{enumerate}[leftmargin=\parindent,itemindent=1mm]

\item\textbf{Deciding on sets.}
\ac{ckcuc} 
decides a set of values $V$ rather than a single value -- solving \ac{cc} as opposed to classical ``selective'' consensus. This simple optimization is highly effective as it substantially increases 
throughput and mitigates the cost of all-to-all message relays. 
\item\textbf{Handling round mismatch and late messages.}
\ac{ckcuc} explicitly considers initial synchronization window and bounded clock drift, rather than a perfectly synchronized round counter~\cite{omission_sync_uniform_consensus:SPAA04} which makes algorithms more compact but is unimplementable in practical systems.
As a result, a process may lag behind, e.g., receive a message before even starting.
\ac{ckcuc} processes explicitly handle messages from processes one round ahead of them by immediately triggering the start of the next round in advance (\autoref{line:ckcapx_roundMismatch2}). Processes also explicitly discard late messages from slow peers (\autoref{line:ckcapx_syncEnforce}) to avoid polluting their $V$ and $locked$ sets, which would otherwise lead to inconsistent decisions.
\item\textbf{Full value set relay to simplify verification.}
Processes in \ac{ckc} always multi-send the full value set $V_i$ instead of just the values learned in the previous round~\cite{omission_sync_uniform_consensus:SPAA04}.
This does not improve efficiency, but it makes the algorithm easier to understand and considerably cheaper to model check, which is what allowed us to verify \ac{ckcuc} with the \apalache{} model checker~\cite{OtoniFKKMOPTK:2026}.
Full-relay keeps the reachable state graph smaller than $new$-relay does, mainly by removing the per-process $new$ variable and the round-boundary action resetting it, as a message payload is simply a projection of the sender's current state.
Furthermore, \ac{ckc} increases the bandwidth consumption \textit{only in the occurrence of failures}, i.e., with more than 2 rounds, as both versions have processes multi-send $V$ in the second round. As failures are extremely rare, \ac{ckc} benefits from verified correctness at little practical cost. 

\end{enumerate}

\subsection{Formal verification}
We specified \ac{ckcuc} in TLA$^+$~\cite{Lamport:2002} and checked its properties using the \apalache{} model checker~\cite{OtoniFKKMOPTK:2026}. All properties of \ac{ckcuc} were shown to hold. The full formal specification and instructions for independent checking are readily available upon request and will be published upon acceptance of the paper.
In addition we provide a detailed correctness pen-and-paper proof of \ac{ckcuc} below.

\subsection{Correctness}

\input{sections/appendix/correctness_parv_raynal.tex}

\section{\ac{ckc}-based \acf{smr}}
This section presents the \ac{tob} abstraction built on top of \ac{ckc} to achieve \ac{smr}, a variant of \ac{ckc} that relays full value sets and a discussion of possible algorithmic optimizations.

\subsection{Building blocks for total order broadcast}

\newacro{fckcuc}[POPUC]{}\acused{fckcuc}


To simplify presentation, we introduce in this section the abstractions that are used as part of the full \ac{tob} specification, and then detail \ac{tob} itself in the next section. To better distinguish between different algorithms, we refer in the following to the latter algorithm from \autoref{alg:apx_ckcuc},  dubbed simply \ac{ckcuc} in the main paper, explicitly as \ac{fckcuc} in the following.

\paragraph{Synchronized start.}
Synchronous algorithms that are round-based, including \ac{fckcuc}, leverage the assumption that all processes start within a fixed time window. In real-world applications, we need a way to enforce this by transitioning from the asynchronous model to the synchronous one. \autoref{alg:syncstart} introduces a \prim{sync-start} primitive for this purpose, using a leader to unequivocally broadcast a starting signal for a round-based synchronous algorithm. A reliable (``perfect'') failure detector $\mathcal{P}$ ensures that a group of correct processes will eventually start.
The maximum synchronization time window $\Delta_W$ is equal to twice the upper-bounded interaction latency $\Delta_I$ of our model. Note that $\mathcal{P}$ (commonly approximated in practice by the use of worst-case timeouts and selective process killing) is needed exclusively for this primitive, whose use is limited to the first phase of the \ac{tob} algorithm, as we will see shortly. Recent work~\cite{fide} demonstrates the feasibility of implementing such $\mathcal{P}$ in datacenters.
\prim{sync-start} uses similar logic as precise clock synchronization (e.g., \acl{ptp}~\cite{ptp_ieee1588}) mechanisms, which can be used instead or in combination with \prim{sync-start} to trigger the start of round-based synchronous algorithm.

\begin{algorithm}[t!]
\caption{\prim{sync-start}. Uses reliable failure detector $\mathcal{P}$. Executed by every $p_i \in \Pi$.}
\label{alg:syncstart}

\nl$leader \gets ~$\prim{min}{($\Pi$)}\;
\nl$started \gets \algval{false}$

\nl\To(\algprim{sync-start}{}:) {
\nl     \eIf{$p_i = leader$ \Kwand $started = \algval{false}$}{
\nl         \prim{send}{(\algval{start})} to 
every $p_j \in \Pi$\;
\nl         $started \gets \algval{true}$
        }{
\nl         \prim{send}{(\algval{start-request})} to $leader$\;
        }
}

\nl\Upon(\prim{recv}{(\algval{start-request})}:){
\nl     \prim{send}{(\algval{start})} to 
every $p_j \in \Pi$\;
\nl     $started \gets \algval{true}$
}

\nl\Upon(\prim{recv}{(\algval{start})}:){
\nl     \prim{start}
}   

\nl\Upon(\algprim[$\mathcal{P}$]{failure}{(p_j)} \Kwand $started = \algval{false}$:){
\nl     $\Pi \gets \Pi \setminus ~\{p_j\}$\;
\nl     $leader \gets ~$\prim{min}{($\Pi$)} \tcp{leader election}
\nl     \algprim{sync-start}{}
}
\end{algorithm}

\paragraph{\Acf{ckcucmi}.}
\acused{ckcucmi}
The second building block 
is an adaptation of \ac{fckcuc}  to account for multiple sequential consensus instances used by the \ac{tob} abstraction.
In \ac{ckcucmi}, a process $p_i$ adds a consensus instance number ($c_i$) to relayed messages. $c_i$ is used to keep track of possible mismatching consensus instances between different processes, e.g., when $p_j$ advances to the next consensus instance before $p_i$. 
Processes proceed with \ac{fckcuc} logic only when they receive messages belonging to the current consensus instance, ignore messages from prior instances, and declare themselves faulty through \algprim{crash}{} when they notice they are a number of instances ``behind''.








\acused{ckctob}

\subsection{Popular total order broadcast (\ac{ckctob})}

The \ac{ckctob} primitive uses \prim{sync-start} and multiple sequential instances of \ac{ckcucmi} to implement \ac{tob}. 
We use the uniform variant of \ac{tob} with a \prim{broadcast} downcall and a \prim{deliver} upcall which satisfies the following properties:

\begin{description}[
font=\textnormal]
    \item[\prop{Validity}:] If a correct process $p_i$ \prim{broadcast}s a message $m$, then $p_i$ eventually \prim{deliver}s $m$.
    %
    \item[\prop{No duplication}:] No message is \prim{deliver}ed twice.
    \item[\prop{No creation}:] If a process \prim{deliver}s a message $m$ with sender $p_i$, then $m$ was previously \prim{broadcast} by process $p_i$.
    \item[\prop{Uniform agreement}:] If a message $m$ is \prim{deliver}ed by some process (correct or not), them $m$ is eventually \prim{deliver}ed by every correct process.
    \item[\prop{Uniform total order}:] Let $m_1$ and $m_2$ be any two messages and suppose $p_i$ and $p_j$ are any two processes (correct or faulty) that \prim{deliver} $m_1$ and $m_2$. If $p_i$ \prim{deliver}s $m_1$ before $m_2$, then $p_j$ \prim{deliver}s $m_1$ before $m_2$.
\end{description}

\begin{algorithm}
\caption{\ac{ckctob}. Executed by every $p_i \in \Pi$. Uses \ac{ckcucmi}.}
\label{alg:ckctob}

\nl$pending, delivered \gets \emptyset$\;
\nl$c_i \gets 0$\;

\nl\To(\algprim{broadcast}{(v)}:) {    
\nl     $pending \gets pending \cup \{v\}$\;
\nl     \If{$c_i = 0$}{
\nl         $c_i \gets c_i + 1$\;    
\nl         \algprim{sync-start}{(\algprim{propose}{(c_i, ~v))}}
}
}

\nl\Upon(\algprim{decide}{(V_i)}:) {
\nl     \ForAll{$v \in  V_i \mid v \ne \square$}{
\nl         \If{$v \notin delivered$}{
\nl         \algprim{deliver}{(v)}\;
\nl         $pending \gets pending \setminus \{v\}$\;
\nl          $delivered \gets delivered \cup \{v\}$
}
}
\nl     $c_i \gets c_i + 1$\;
\nl     \eIf{$pending = \emptyset$}{
\nl         \algprim{propose}{(c_i, ~\square)}
        }{
\nl         \algprim{propose}{(c_i, ~\prim{next}{(pending)})}
        }
}

\end{algorithm}

\paragraph{Algorithm.}
\autoref{alg:ckctob} outlines \ac{ckctob}. It uses a series of monotonically-increasing, back-to-back consensus instances to order a list of values. 
\ac{ckctob} uses the \algprim{sync-start}{} primitive at the first \prim{broadcast} call to ensure that all processes \ac{ckcucmi}-\prim{propose} within a fixed time window $\Delta_W$. 
Future
\prim{broadcast} events add values to the $pending$ set which acts as a buffer for \prim{broadcast} calls in quick succession (line 4). The \prim{next} primitive is used to randomly pick a value from $pending$ to 
propose in the following consensus instance (line 17). If there are no pending values, the $\square$ value is proposed 
instead to serve as a heartbeat for other processes (line 16). This is a key mechanism to avoid false positives, since processes not participating in a round are considered to be faulty in our model with omission failures. Furthermore, heartbeats maintain the synchronous operation of the algorithm in the temporary absence of a message: running consensus instances independently, i.e., not back-to-back, would require to \prim{sync-start} every consensus instance, leading to expensive message relays.
Heartbeat values $\square$ are filtered before delivery 
 (line 10).
 \textbf{forall} to be deterministic and cycle through elements of the $V_i$ set in the same order for all processes which is easily achieved, e.g., by using message and sender \textsc{id}s. 
In failure-free executions, all processes execute a series of consensus instances, each of which take exactly 2 rounds thanks to the early stopping feature of \ac{ckcucmi} ($t+2$ rounds as lower bound). Upon failures, some processes might take additional rounds, leading to overlapping consensus instances. The \ac{ckcucmi} abstraction (\autoref{alg:ckcucmi}) deals with this issue by ensuring that only instance-aligned processes participate in \ac{ckctob}.

%% file: sections/appendix/correctness_parv_raynal.tex
We prove that \autoref{alg:apx_ckc} satisfies \prop{Validity}, \prop{Termination} and \prop{Uniform agreement} in a system of $N = 2t+1$ processes, of which at most $t$ are faulty. The proof follows the structure of Parv{\'e}dy \& Raynal~\cite{omission_sync_uniform_consensus:SPAA04}.

\paragraph{Notation.} A process \emph{closes} round $r$ when it runs the round change routine (lines~\ref{line:ckcapx_suspectUpdate} to \ref{line:ckcapx_multiSend}), either on \prim{timeout} or through the catch-up of \autoref{line:ckcapx_roundMismatch2}. It \emph{crashes at}, \emph{becomes locked at}, or \emph{decides at} $r$ when it executes \autoref{line:ckcapx_crash}, \autoref{line:ckcapx_lockedSelfAdd}, or one of lines~\ref{line:ckcapx_decideCrash} and~\ref{line:ckcapx_decidebound} while closing $r$. We write $V_i[r]$, $suspect_i[r]$ and $locked_i[r]$ for the values of the corresponding variables at the end of that routine, and $locked_i^{-}[r]$ for the value of $locked_i$ when \autoref{line:ckcapx_lockedCheck} is evaluated, i.e. after all the merges of round $r$ but before $p_i$ may add itself. If $p_i$ decides at $r$ we let $V_i[r'] = V_i[r]$ for every $r' > r$, since $V_i$ is no longer updated. We also write: 
\begin{itemize} 
    \item $M_i[r]$ for the set of processes whose round $r$ message $p_i$ merged; \item $stopped[r]$ for the processes that crashed at some round $r' \le r$ 
    \item $good[r] = \Pi \setminus stopped[r]$ for those still running or already decided; 
    \item $V[r] = \bigcup_{p_i \in good[r]} V_i[r]$ for the values known by the processes of $good[r]$. 
\end{itemize} 

A process \emph{knows} $v$ at $r$ if $v \in V_i[r]$, and \emph{learns} $v$ at $r$ if $v \in V_i[r] \setminus V_i[r-1]$. The \emph{round $r$ message} of $p_i$ is the one it sends at \autoref{line:ckcapx_multiSend} while closing round $r-1$; it carries $V_i[r-1]$ and $locked_i[r-1]$.

\paragraph{Preliminary observations.}

\begin{observation}[Asynchronous round alignment]
    \label{obs:align} A message is merged only by a process that is in the round the message is tagged with. Hence $p_i$ merges the round $r$ message of $p_j$ only while it is in round $r$, and that message carries $V_j[r-1]$ and $locked_j[r-1]$. 
\end{observation} 
\begin{proof} 
    \autoref{line:ckcapx_syncEnforce} merges only when $r_j = r_i$. A message with $r_j = r_i + 1$ first triggers the round change of \autoref{line:ckcapx_roundMismatch2}, after which $r_i = r_j$ and the test succeeds; messages with $r_j < r_i$ or $r_j > r_i + 1$ are discarded. Note that this is a property of the protocol text alone: it holds whatever the message delays. 
\end{proof}

\begin{assumption}[Synchronous rounds]
    \label{ass:sync} 
    The timeout duration $\Delta_I + \Delta_W + \Delta_D$ is chosen so that the round $r$ message of a correct process reaches every correct process before the latter closes round $r$. In particular correct processes never lag behind one another by more than one round, so their messages are never discarded as stale.
\end{assumption} 

Note that safety properties (\autoref{th:validity} and \autoref{th:agreement}) rely only on \autoref{obs:align} and not assumed time bounds. 
Timing assumptions (\autoref{ass:sync}) are used only for liveness (\autoref{th:termination}), through the following property:  

\begin{observation}[Active exclusion]\label{obs:suspect} 
    
    If $p_i$ closes round $r$, then $suspect_i[r] = \Pi \setminus (\{p_i\} \cup locked_i^{-}[r] \cup M_i[r])$. In particular, if $locked_i^{-}[r] = \emptyset$ then $p_i$ merged the round $r$ message of every process it does not suspect at $r$. Moreover, if $p_j \in suspect_i[r]$ and $p_j \notin locked_i$, then $p_i$ merges no further message of $p_j$ and sends it none. 
\end{observation} 
    
\begin{proof} 
        $roundLiveset$ is reset to $\{p_i\} \cup locked_i$ when round $r-1$ is closed (\autoref{line:ckcapx_livesetReset}) and then grows, at every merge, with the sender and with the current value of $locked_i$ (\autoref{line:ckcapx_livesetUpdate2}); since $locked_i$ only grows, the union of these values is $locked_i^{-}[r]$, which subsumes $locked_i[r-1]$. \autoref{line:ckcapx_suspectUpdate} takes the complement. The last claim is the guard of \autoref{line:ckcapx_recv} and the send set of \autoref{line:ckcapx_multiSend}. The exception on $locked_i$ plays the role of the one used in~\cite{omission_sync_uniform_consensus:SPAA04}: a process that decides early stops sending, and without the exception a correct process could accumulate more than $t$ suspicions and quit at \autoref{line:ckcapx_crash} without deciding. 
\end{proof}

\begin{observation}[Relay]\label{obs:relay} If $p_i$ merges the round $r$ message of $p_j$, then $V_j[r-1] \subseteq V_i[r]$ and $locked_j[r-1] \subseteq locked_i^{-}[r]$ (lines~\ref{line:ckcapx_vrecv} and~\ref{line:ckcapx_lockedUpdate}). \end{observation}

\begin{observation}[Locking]\label{obs:locked} $locked_i$ never shrinks and a process becomes locked at most once. Moreover: \begin{enumerate} \item if $locked_i^{-}[r] \ne \emptyset$ then $p_i$ is locked at $r$ at the latest; equivalently, if $p_i$ is locked at no round $\le r$ then $locked_i^{-}[r'] = \emptyset$ for every $r' \le r$; \item $p_k \in locked_i^{-}[r]$ implies that $p_k$ became locked at some round $\le r-1$, unless $p_k = p_i$. \end{enumerate} \end{observation} \begin{proof} The first point is the test of \autoref{line:ckcapx_lockedCheck}. For the second, a process enters a lock set either by adding itself (\autoref{line:ckcapx_lockedSelfAdd}) or by relay of a set that already contained it (\autoref{line:ckcapx_lockedUpdate}), and by \autoref{obs:align} a relay carries $locked_j[r-1]$, hence costs at least one round. \end{proof}

\begin{observation}[Survivors keep a majority]\label{obs:majority} 
    If $p_i \in good[r]$ then $|suspect_i[r]| \le t$, so $R_i[r] = \Pi \setminus suspect_i[r]$ has at least $t+1$ elements. 
\end{observation}

\paragraph{Safety.}

\begin{theorem}[Validity]\label{th:validity} If a process \prim{decide}s $V$, then every $v \in V$ was \prim{propose}d. \end{theorem} \begin{proof} $V_i$ is initialised with the proposal of $p_i$ (\autoref{line:ckcapx_proposevalupdate}) and grows only with sets received from other processes (\autoref{line:ckcapx_vrecv}). By induction on rounds, every message carries proposed values only. The decided set is $V_i$. \end{proof}

\begin{lemma}[Local knowledge grows]\label{lem:local} 
    If $p_i \in good[r]$ and $r' \le r$, then $V_i[r'] \subseteq V_i[r]$. 
\end{lemma} 
\begin{proof} No value is ever removed from $V_i$. \end{proof}

\begin{lemma}[Global knowledge shrinks]\label{lem:global} 
    If $r' \le r$, then $V[r] \subseteq V[r']$. 
\end{lemma} 
\begin{proof} 
    It suffices to prove $V[r] \subseteq V[r-1]$. Let $v \in V_i[r]$ with $p_i \in good[r]$. If $v \in V_i[r-1]$ then the condition holds due to $good[r] \subseteq good[r-1]$. Otherwise $p_i$ learns $v$ at $r$ from the round $r$ message of some $p_j$, so $v \in V_j[r-1]$ (\autoref{obs:align} and~\ref{obs:relay}). As $p_j$ sent that message while closing $r-1$, it did not crash at $r-1$, hence $p_j \in good[r-1]$ and $v \in V[r-1]$. 
\end{proof}

\begin{lemma}[Lock implies knowledge completeness]\label{lem:lockedknows} 
    Assume no process decides before $r$, and let $p_i$ become locked at $r$. Then $V_i[r] = V[r]$. 
\end{lemma} 

\begin{proof} 

    The locking condition is $r_i > suspect_i $ \textbf{or} $locked_i \ne \emptyset$ (\autoref{line:ckcapx_lockedCheck}), hence we have two cases:

    \begin{itemize}[leftmargin=1.1cm]

        \item[\emph{Case (i)}]: $locked_i^{-}[r] = \emptyset$, hence $r > |suspect_i[r]|$. 
        
        Consider a process $p_j \in good[r]$ and a value $v \in V_j[r]$. We show that $v \in V_i[r]$. Assume $r'$ as the round at which $p_j$ learns $v$. Since $p_j \in good[r]$, then $r' \le r$. Specifically, we have two subcases:

        \begin{itemize}
            
        \item \emph{Subcase $r' = r$.} Assume by contradiction that $v \notin V_i[r]$. By \autoref{lem:local}, $p_i$ does not know $v$ at $r-1$. As processes always forward the full value set, it follows that $v$ has passed through a chain of $r'$ distinct processes before reaching $p_j$. As $v \notin V_i[r]$, $p_i$ must suspect at least $r$ processes, i.e., $r \le |suspect_i[r]|$, contradicting $r > |suspect_i[r]|$.

        \item \emph{Subcase $r' = r-1$.} Assume by contradiction as above. Here $p_i$ must suspect $r-1$ processes plus $p_j$, from which it would learn $v$ at $r$ (\autoref{obs:relay}). Hence $r \le |suspect_i[r]|$, again contradicting $r > |suspect_i[r]|$.

        \item \emph{Subcase $r' \le r-2$.} By \autoref{obs:majority} both $R_i = \Pi \setminus suspect_i[r]$ and $R_j = \Pi \setminus suspect_j[r]$ have at least $t+1$ elements, and $N = 2t+1$, so there is some $p_k \in R_i \cap R_j$. If $p_k = p_j$ then by \autoref{obs:relay} $p_j$ must have relayed $v$ to $p_i$ in later rounds, so $v \in V_i[r]$, contradicting the case assumption. 
        Otherwise, $p_k$ is not locked at any round $\le r-1$ because $p_k \in R_i$, so $p_k \notin locked_j^{-}[r]$ (\autoref{obs:locked}) and $p_k \in M_j[r]$: $p_j$ merged the round $r$ message of $p_k$, so $p_k$ sent to $p_j$ at $r$ and $p_j \notin suspect_k[r-1]$. As $locked_k^{-}[r-1] = \emptyset$, \autoref{obs:suspect} gives $p_j \in M_k[r-1]$, so $p_k$ merged the round $r-1$ message of $p_j$, which carries $V_j[r-2] \ni v$. Hence $p_k \in suspect_i[r]$, contradicting $p_k \in R_i$.

        \end{itemize}

        \item [\emph{Case (ii)}]: $p_k \in locked_i^{-}[r]$ 
        
        $p_k$ must have locked through \textit{Case (i)} or received a lock from some $p_j$ that locked through \textit{Case (i)}. In either case, $V_k[r-1] = V[r-1]$, $V_i[r] \supseteq V_k[r-1]$ by \autoref{obs:relay}. By \autoref{lem:global}, $V[r] \subseteq V[r-1]$, so $V_i[r] \supseteq V[r]$. The reverse inclusion is always true, so $V_i[r] = V[r]$.
        
    \end{itemize}

\end{proof}

\begin{lemma}\label{lem:frozen} Let $p_i \in good[r]$ be locked at some round $y \le r$, and assume no process decides before $y$. Then $V_i[r] = V[r] = V[y]$. 
\end{lemma} 

\begin{proof} $V[y] = V_i[y] \subseteq V_i[r] \subseteq V[r] \subseteq V[y]$, by \autoref{lem:lockedknows}, \autoref{lem:local}, the definition of $V[r]$, and \autoref{lem:global}. \end{proof}

\begin{lemma}\label{lem:decidelocked} A process is locked when it decides. \end{lemma} \begin{proof} \autoref{line:ckcapx_decideCrash} is the \textbf{elif} branch of \autoref{line:ckcapx_selfLockedCheck}, so it is reached only when $p_i \in locked_i$ already. A process reaching \autoref{line:ckcapx_decidebound} closes round $t+1$ without crashing, so $|suspect_i[t+1]| \le t < t+1$; the test of \autoref{line:ckcapx_lockedCheck} then holds and $p_i$ is locked at $t+1$ at the latest, before \autoref{line:ckcapx_decidebound} is evaluated. \end{proof}

\begin{lemma}[Locks carry knowledge]\label{lem:lockcarries}
If $p_k \in locked_i[r]$, then $p_k$ became locked at some round $y \le r$ and
$V_k[y] \subseteq V_i[r]$.
\end{lemma}
\begin{proof}
By induction on $r$; lock sets are initially empty. There are three ways for $p_k$ to be in
$locked_i[r]$. If $p_k = p_i$ added itself at $r$ (\autoref{line:ckcapx_lockedSelfAdd}), then
$y = r$ and the claim is trivial. If $p_k$ was already in $locked_i[r-1]$, the induction
hypothesis and \autoref{lem:local} give $V_k[y] \subseteq V_i[r-1] \subseteq V_i[r]$. Otherwise
$p_k$ arrived with the round $r$ message of some $p_j$, which carries $locked_j[r-1]$ and
$V_j[r-1]$ (\autoref{obs:align}). Applying the induction hypothesis to $p_j$ at $r-1$ yields
$y \le r-1$ and $V_k[y] \subseteq V_j[r-1]$, and $V_j[r-1] \subseteq V_i[r]$ by
\autoref{obs:relay}.
\end{proof}

\begin{theorem}[Uniform agreement]\label{th:agreement}
No two processes decide different sets of values. This holds for any number of faulty
processes.
\end{theorem}
\begin{proof}
Let $r$ be the first round at which a process decides, and let $p_i$ decide at $r$. By
\autoref{lem:decidelocked} $p_i$ is locked at some round $\le r$, so $p_i$ decides
$V_i[r] = V[r]$ (\autoref{lem:frozen}). The same argument applies to any process deciding at
$r$, so all decisions taken at $r$ are identical.

Let now $p_j$ decide at $r'' > r$. Every process of $good[t+1]$ decides while closing $t+1$
(\autoref{line:ckcapx_decidebound}), so $r'' \le t+1$ and $r \le t$; at round $r \le t$ the test
of \autoref{line:ckcapx_decidebound} fails, so $p_i$ decided at
\autoref{line:ckcapx_decideCrash} with $|locked_i| > t$, i.e. $|locked_i[r]| \ge t+1$. Since
$p_j \in good[r'']$, \autoref{obs:majority} gives $|R_j[r'']| \ge t+1$, and $N = 2t+1$, so there
is some $p_k \in locked_i[r] \cap R_j[r'']$. By \autoref{lem:lockcarries}, $p_k$ became locked at
some $y \le r$, and by \autoref{lem:lockedknows} and \autoref{lem:global},
$V[r] \subseteq V[y] = V_k[y]$.

It remains to show $V_k[y] \subseteq V_j[r'']$. By \autoref{obs:suspect},
$R_j[r''] = \{p_j\} \cup locked_j^{-}[r''] \cup M_j[r'']$. If $p_k = p_j$, then
$V_k[y] \subseteq V_j[r'']$ by \autoref{lem:local}. If $p_k \in locked_j^{-}[r'']$, the same
follows from \autoref{lem:lockcarries}. If $p_k \in M_j[r'']$, its round $r''$ message carries
$V_k[r''-1] \supseteq V_k[y]$ (\autoref{obs:align} and \autoref{lem:local}), which $p_j$ merges.
In all cases $V[r] \subseteq V_j[r''] \subseteq V[r''] \subseteq V[r]$, so $p_j$ decides $V[r]$
as well.

The proof uses only $N = 2t+1$, the bound $|suspect_i| \le t$ enforced at
\autoref{line:ckcapx_crash}, and the early decision threshold $|locked_i| > t$. No assumption on
the actual number of faulty processes is made.
\end{proof}

\paragraph{Liveness.}

\begin{lemma}\label{lem:trust} No correct process ever suspects a correct process. Consequently $|suspect_i[r]| \le f$ for every correct $p_i$ and every $r$, and no correct process crashes. \end{lemma} \begin{proof} By induction on $r$; assume that up to round $r-1$ no correct process suspected a correct one, so no correct process crashed. Let $p_i$ and $p_c$ be correct. While both are running there is no omission between them, and by \autoref{ass:sync} the round $r$ message of $p_c$ reaches $p_i$ before $p_i$ closes $r$, so $p_c \in M_i[r]$. Assume instead that $p_c$ decided at some $z < r$ and stopped. By \autoref{lem:decidelocked} it is locked at some $y \le z$. If $y < z$, its round $y+1 \le z$ message carries $p_c \in locked_{p_c}[y]$ and is merged by $p_i$, so $p_c \in locked_i$ from then on and the exception of \autoref{obs:suspect} applies. If $y = z$, then $p_c$ decided at \autoref{line:ckcapx_decidebound} while closing $t+1$, so $z = t+1$, $p_i$ also decides while closing $t+1$, and no round $r > z$ is closed. In all cases $p_c \notin suspect_i[r]$, hence $|suspect_i[r]| \le f \le t$ and the test of \autoref{line:ckcapx_crash} never holds for $p_i$. 

\end{proof}

\begin{theorem}[Termination]\label{th:termination} 
    Every correct process decides. \end{theorem} 
    
\begin{proof} 
    By \autoref{lem:trust} a correct process never crashes at \autoref{line:ckcapx_crash}, and by \autoref{ass:sync} its timer makes it close one round after the other. It therefore closes round $t+1$ unless it decided earlier, and decides at \autoref{line:ckcapx_decidebound}. 
\end{proof}